\pdfoutput=1 
\documentclass{article}
\usepackage{graphicx, float, caption, subcaption} 

\usepackage{amsthm}
\usepackage{amsfonts}
\usepackage{amsmath, amssymb}
\newtheorem{theorem}{Theorem}[section]

\newtheorem{lemma}[theorem]{Lemma}

\theoremstyle{definition}
\newtheorem{definition}{Definition}[section]

\usepackage[numbers]{natbib}

\usepackage{tikz}
\usetikzlibrary{shapes.geometric, arrows}
\usetikzlibrary{positioning}

\tikzstyle{box} = [rectangle, draw, minimum width=3cm, minimum height=1cm, text centered]
\tikzstyle{arrow} = [thick,->,>=stealth]

\usepackage{xcolor}
\usepackage{soul} 
\usepackage{ulem}

\newcommand{\R}{{\mathbb{R}}}
\newcommand{\Z}{{\mathbb{Z}}}
\newcommand{\T}{{\mathbb{T}}}
\newcommand{\rd}{\text{d}}
\newcommand{\N}{{\mathbb{N}}}

\title{Persistence of \(n\)-Species Lotka--Volterra Models with Periodic Pulses}
\author{Eleanor Courcelle\thanks{Department of Mathematics,
Oregon State University,
Corvallis, OR 97331}\and Jane Shaw MacDonald\thanks{Arcadis, Suite 100, 1285 West Pender Street, Vancouver, British Columbia, V6E 4B1, Canada, Email: jane.madonald2@arcadis.com } \and
Swati Patel\thanks{Department of Mathematics,
Oregon State University,
Corvallis, OR 97331\\
Email: patelswa@oregonstate.edu}}

\date{}

\begin{document}
\maketitle

\begin{abstract}
Periodic impulsive interventions arise naturally in the management of biological populations, including chemotherapy, pesticide application, and infectious-disease treatment. We develop general conditions for permanence in \(n\)-species population models subject to periodic multiplicative pulse disturbances. Our main result provides a sufficient condition for permanence in terms of weighted long-term growth rates on a Morse decomposition of the extinction set, explicitly separating the contributions of continuous population dynamics from those of the periodic pulse. To establish this result, we transform the impulsive system into an associated autonomous continuous-time dynamical system and use this correspondence to extend classical permanence theory to periodically pulsed models. We further show that the same conditions imply robust permanence under sufficiently small perturbations to the continuous dynamics, pulse period, and pulse effects. We illustrate the framework with two Lotka--Volterra models motivated by biological control: competition between chemotherapy-sensitive and chemotherapy-resistant cancer cells, and integrated control of an agricultural pest using pesticides and parasitoids. These examples demonstrate how intervention frequency and intensity interact with underlying ecological interactions to determine whether populations coexist or are excluded. Our results provide a general framework for analyzing persistence in ecological systems subject to repeated discrete disturbances.
\end{abstract}

\section{Introduction}
Periodic, impulsive control interventions are common in ecological population management. For example, in cancer treatment,  chemotherapy is often administered in discrete cycles, with each dose sharply reducing the tumor-cell population followed by a treatment-free interval during which both tumor cells and healthy tissues recover \cite{lakmeche_chemo_periodic_2000, ren2017tumour, diovidio_pulsed_cell_pop_2024}. In agriculture, crop pests and pathogens are frequently controlled through repeated pesticide applications at prescribed intervals \cite{mailleret_pulsed_biocontrol_2009, tang_periodic_lv_2002, mermer2021timing}.  For example, field trials in hot pepper production have used periodic pesticide applications every 7–10 days to suppress anthracnose, tobacco budworm, and other pests and diseases \cite{kim2013three}. Similarly, infectious-disease control programs use mass drug administration (MDA) to produce population-wide pulses of treatment: annual or semiannual rounds of antiparasitic drugs are used to suppress transmission of diseases such as trachoma and lymphatic filariasis, with pathogen prevalence or parasite burden increasing between successive treatment rounds \cite{gao2017mass, WHO2006preventive, patel_spectral_2024, patel2025anthelmintic}. These examples share a natural impulsive structure in which population dynamics proceed continuously between interventions but undergo abrupt changes at regularly scheduled treatment times.
In this work, we examine the properties that lead to persistence of ecological communities or populations under the influence of periodically pulsed disturbances.

Impulsive differential equations (IDEs) are commonly used to model such pulsed interventions of ecological community or population dynamics.  IDEs share three defining components—a set of continuous-time ordinary differential equations, criteria governing when the moments of instantaneous disturbances occur, and an impulse equation describing how the system is transformed by the treatment pulse \cite{stamova_overview_ides_2016}. These are appropriate when the pulsed disturbance has a relatively fast impact on the population compared to the time scale of the ecological life cycle. Hence, the pulse is modeled as being instantaneous. We concentrate on systems with impulse equations that impose a multiplicative-scaling over regular time intervals. These model assumptions are intuitive for parameterizing many relevant biological systems, as efficacy rates for human-implemented control methods are often reported on a per capita basis.  

IDEs are commonly used to model and examine optimal and effective control in ecological systems, including in agriculture \cite{mailleret_pulsed_biocontrol_2009, tang_periodic_lv_2002, tang2010optimum}, cancer \cite{diovidio_pulsed_cell_pop_2024}, and infectious diseases \cite{agur1993pulse, shulgin1998pulse, donofrio2002pulse, patel2025anthelmintic}
Most of this work focuses on identifying strategies towards elimination goals of a particular harmful species. We take a different perspective and instead examine conditions that enable multiple populations to persist when eradication is infeasible or undesirable. This setting is particularly relevant in applications where biological constraints on treatment intensity, limited intervention resources, or the risk of resistance evolution preclude aggressive control. For example, chemotherapy doses are constrained by patient tolerance \cite{wong2015dose}, while infectious-disease interventions may be limited by treatment resources \cite{hansen2011optimal}. Furthermore, repeated application of control agents can impose strong selection and increase the risk of evolving resistance in agricultural pests \cite{karlssongreen2020making}, cancer cell populations \cite{foo2014evolution}, and infectious pathogens \cite{andersson2010antibiotic}. In such settings, coexistence with beneficial or treatment-sensitive populations may provide an alternative management objective \cite{gatenby_cost_of_resistance_2009, whelan2020resistance}.

Permanence theory is a formal mathematical framework that developed in the 80s to examine coexistence of different classes of ecological models \cite{hutson_overview_persistence_1999}.
Heuristically, permanence is a property of an ecological model that ensures that the species densities stay bounded and eventually remain above a minimum threshold.  
To date, a rich theory has been developed \cite{gard1987uniform, sigmund1984permanence, garay_hofbauer_2003, patel_schreiber_2018, HofbauerSchreiber2022InvasionGraphs} 
to identify sufficient conditions for permanence in $n$-species Kolmogorov systems where population growth can be described by continuous ordinary differential equations of the form
\begin{equation}\label{kolm}
    \frac{dx_i}{dt} = x_i f_i(x).
\end{equation}
Establishing permanence in ODE models requires showing that trajectories starting in the interior (with positive density) remain uniformly bounded away from extinction sets (where one species is zero). Indeed, in the development of the theory, different approaches and frameworks have been used, invariant-set decompositions and average Lyapunov functions \cite{garay_hofbauer_2003, law_permanence_1996, patel_schreiber_2018}, or, more recently, invasion graphs \cite{HofbauerSchreiber2022InvasionGraphs, spaak2023building}.

Conditions for permanence of impulsive differential models have been examined. For example, \cite{ballinger_persistence_of_general_ides_1997} studied permanence but their approach requires the construction of a suitable Lyapunov function to prove permanence criteria is met. Other work, such as \cite{jin_persistence_2005} is system- or dimension-specific. Recently, Schreiber provided fairly general permanence results for IDEs, which he termed ``flow-kick systems'' \cite{schreiber_flowkick_2025}. 

Here, we provide conditions for robust permanence (as was conjectured in \cite{schreiber_flowkick_2025}) for a class of Kolmogorov-type IDEs, with periodic impulses. Robust permanence is a stronger notion than permanence that requires the property of permanence is robust to appropriately defined perturbations of the model equations \cite{garay_hofbauer_2003}. Furthermore, the method we use to construct our proof is different than from \cite{schreiber_flowkick_2025} and thus, may provide a useful framework for analysis of other periodic impulsive systems. 
Specifically, we present a technique for transforming periodic impulsive equations to bridge with autonomous dynamical systems theory. 
Our main contribution is extending results from Garay and Hofbauer \cite{garay_hofbauer_2003} to impulsive systems. Their results, which generate an explicit condition for permanence of (\ref{kolm}), showed that such systems are robustly permanent whenever solutions contained in a decomposition of the invariant set on the extinction set exhibit positive, weighted-average growth rates. We will provide a similar formulation in our main theorem.

Our results are organized as follows. First, in Section \ref{sect:setup} we establish model preliminaries and assumptions. We also define two dynamical flows to separate the continuous and impulsive behavior of our general IDE model. In Section \ref{sect:results}, we state our main permanence finding and give applications to specific example systems, including an alternate proof of permanence results in the two-dimensional Lotka-Volterra model from \cite{jin_persistence_2005}. We conclude in Section \ref{sect:discussion} by considering implications and potential directions for future work.

\section{Set Up}\label{sect:setup}
\subsection{Model}
Consider a system of $n$ species, whose population density at time $t$ is represented by the population vector $x(t)=(x_1(t), x_2(t),...,x_n(t)) \in \R^n_+$. We will assume that when undisturbed the per capita growth rate of species $i$ can be described by an autonomous, locally-Lipschitz continuous function $f_i(x)$ for $i = 1,\ldots,n$. Suppose that at $\tau$-periodic time intervals, the system is regularly disrupted by a pulse event affecting each species $i$ instantaneously by some scale factor $h_i > - 1$. Then combined, the dynamics of the system can be be modeled by the impulsive differential equations
\begin{equation} \label{eqxn:pulsed_model}
    \begin{split}
        \frac{dx_i}{dt} = x_i(t) f_i(x) ,\;  i = 1 \ldots n\; ,\ t \neq k\tau; k \in \Z_{+}\\
        x_i(k\tau^+) = (1+h_i)x_i(k\tau^-)
    \end{split}
\end{equation}
where $x_i(k\tau^-)$ and $x_i(k\tau^+)$ denote the left-hand and right-hand limits of $x_i$ as $t$ approaches the pulse time at $t=k\tau$. Notice here that if $h_i\in(-1,0)$ then the pulse disturbance has a negative impact on the populations dynamics, which is appropriate in scenarios of population control (see our examples). On the other hand, $h_i>0$ is appropriate, for example, when there are pulsed birth events \cite{lewis2012spreading, vasilyeva2016spread}.

In this work, we are interested in understanding the conditions on System (\ref{eqxn:pulsed_model}) that guarantee the coexistence of all species in the sense of permanence. While biological applications have led to many different mathematical interpretations of this problem, we frame our results in terms of the property of system permanence as stated below. 

\begin{definition} [Permanence] \label{def:permanence}
System (\ref{eqxn:pulsed_model}) is permanent if there exists a $\beta > 0$ such that for any initial condition with $x_i(t_0) > 0$ and associated solution $x(t)$, we have
    \begin{equation*}
        \frac{1}{\beta} < \liminf_{t \to \infty} x_i(t) < \limsup_{t \to \infty} x_i(t) < \beta
    \end{equation*} for all $i$.
\end{definition}

Permanence behavior has been well studied in the unpulsed analog of System (\ref{eqxn:pulsed_model}), i.e. with $h_i = 0$ for all $i$. We generalize this theory to the $n$-species, pulsed System (\ref{eqxn:pulsed_model}) with $h_i > -1$.

\subsection{Terminology and Assumptions}
We begin by defining some useful concepts and terminology from dynamical systems, which we use in our main proofs. 

\begin{definition}[Dynamical system]
Given a state space $X$ and a time set $\mathbb{T}$, a map $\Phi: X\times \mathbb{T} \rightarrow  X$ is a dynamical system if 
\begin{enumerate}
    \item Identity Property: For all $x \in X$, $\Phi(x,0) = x$.
    \item Composition/Semigroup Property: For all $x \in X$ and $t, s \in \mathbb{T}$, $\Phi(x,t+s) = \Phi(\Phi(x,s),t)$
\end{enumerate}
\end{definition}

\subsubsection{Key Dynamical Systems}
Our impulse differential equations do not directly define a dynamical system (they do not satisfy the second property).  Rather they define a non-autonomous dynamical system and can be studied in the framework of skew-product flows \cite{zhao2001uniform, mierczynski2004uniform}.  However, we avoid this by defining two key dynamical systems associated to our impulse differential equation.

The first dynamical system captures the continuous dynamics of System (\ref{eqxn:pulsed_model}). That is, we let $\Phi: \R_+^n \times \R \rightarrow \R_+^n$ be the dynamical system generated by the continuous flow, i.e., 
\begin{equation}\label{contflow}
    \frac{dx_i}{dt} = x_if_i(x)
\end{equation}

The second is associated to the discrete pulses of System \ref{eqxn:pulsed_model}. We define a discrete map for $ x\in\mathbb{R}^n_+$ as follows:
\begin{equation}\label{discretemap}
x_{k+1}= r(\Phi(x_k, \tau))   
\end{equation}
where $r:\mathbb{R}^n_+ \rightarrow \mathbb{R}^n_+$ multiplies component $i$ with $1+h_i$, capturing the impact of the pulse disturbance. We let $\pi$ be the resulting discrete dynamical system. This can be viewed as the Poincare map of the impulse differential equations at the time right after the pulse disturbance. A summary of the relevant dynamical systems is in Table \ref{table}.  
Notice that there is a natural analogous notion of permanence for (\ref{discretemap}), which implies permanence of System (\ref{eqxn:pulsed_model}).  There is a body of work on permanence of discrete dynamical systems \cite{kon2004permanence, salceanu2009lyapunov}, but we will only use Equation (\ref{discretemap}) to identify a Morse decomposition and state our main results. Instead our proof relies on tools for permanence in ordinary differential equations.

\subsubsection{Key Terminology}
Next, we define key terminology. Let $\Phi$ be a given dynamical system on state space $X$ with time set $\mathbb{T}$. 
For any set $A\subset X$ and $T\subset \mathbb{T}$, $\Phi(A,T)= \{\Phi(x,t)| t\in \mathbb{T}, x\in A\}$. 
The \emph{$\omega_\Phi$-limit set} and  \emph{$\alpha_\Phi$-limit set} of a set $A\subset X$ are $$\omega_\Phi(A) = \cap_{t\geq 0}\overline{\Phi(A,[t,\infty))}$$ and $$\alpha_\Phi(A) = \cap_{t\leq 0}\overline{\Phi(A,(-\infty,t))},$$ respectively. Equivalently, for our purposes, 
$$\omega_\Phi(A) = \{y\in X | \exists \text{ an increasing sequence } t_j \rightarrow \infty \text{ and }  x_j \in A \text{ such that } \Phi(x_j, t_j)\rightarrow y \}$$
and 
$$\alpha_\Phi(A) = \{y\in X |\exists \text{ a decreasing sequence } t_j \rightarrow -\infty \text{ and }  x_j \in A \text{ such that } \Phi(x_j, t_j)\rightarrow y \}.$$ The proof of this equivalence can be found in \cite{chicone_odes_2006} under Proposition 1.167.
A set $\Gamma \subset X$ is \emph{invariant} if $\Phi(\Gamma, \T) = \Gamma$. 
We say a compact invariant set $\Gamma$ is \emph{isolated} if there exists a neighborhood $U$ of $\Gamma$ such that $U$ does not contain any invariant sets except those contained in $\Gamma$.  

We next introduce Morse decompositions, which are partitionings of invariant sets and useful in our context for decomposing the boundary.  The connections with permanence theory were made in \cite{garay1989uniform}.

\begin{definition}[Morse Decomposition]
    Given a dynamical system $\Phi$ the family of sets $\mathcal{M}_\Phi = \{M_1, M_2, \ldots, M_l\}$ is a Morse decomposition for a compact invariant set $\Gamma$ if
    \begin{enumerate}
        \item $M_1, M_2, \ldots, M_l$ are pairwise disjoint, isolated invariant compact sets, and
        \item For all $z \in \Gamma \setminus \bigcup^l_{k=1} M_k$, there exists $1 \leq i < j \leq l$ such that $\omega(z) \subset M_i$ and $\alpha(z) \subset M_j$.
    \end{enumerate}
\end{definition}

\begin{table}[]
    \centering
    \begin{tabular}{c|c|c}
        continuous flow &  discrete map & embedded autonomous flow \\
        \hline 
        & & \\
        
        $\Phi: \R^n_+ \times \R \rightarrow \R^n_+ $& $\pi: \R^n_+ \times \N \rightarrow \R^n_+ $ & $\phi: \R^n_+ \times S^1 \times \R \rightarrow \R^n_+ \times S^1 $\\

         & & \\
eqns (\ref{contflow}) & eqns (\ref{discretemap}) & eqns (\ref{auto}-\ref{auto2}) \\
    \end{tabular}
    \caption{Table of three dynamical systems}
    \label{table}
\end{table}

\section{Results}\label{sect:results}
\subsubsection{Assumptions}
Before stating our first main result, we make two standing assumptions on our impulse differential equation (\ref{eqxn:pulsed_model}) by considering the associated continual flow and the discrete map:
\\

\noindent \textbf{A1:} $x_if_i$ are locally Lipschitz functions, and\\
\noindent \textbf{A2:} there exists a compact set $Q \subset \R^n_+$ such that $\pi(Q, \N) \subset Q$ and for any $x\in \R^n_+$ there is an $n>0$ such that $\pi(x, n) \in Q$

The first assumption is a regulatory assumption to ensure existence and uniqueness of solutions. The second ensures the \emph{dissipative} property that populations will not grow unboundedly.

\subsubsection{Permanence Result}

To state our main result, we define the extinction set $E = \{x\in \R^n_+ | \Pi_{i=1}^n x_i = 0 \}$, i.e., the set in which at least one species has density of 0.  
Given the above assumptions, we define the \emph{global attractor} of the discrete dynamical system $\pi$ as $$\Gamma_\pi = \omega_\pi(Q).$$
While the choice of the set $Q$ satisfying \textbf{A2} is not unique, $\omega_\pi(Q)$ is unique.  That is, for two sets $Q_1$ and $Q_2$ satisfying \textbf{A2}, $\omega_\pi(Q_1)=\omega_\pi(Q_2)$. 

Our main permanence result is

\begin{theorem} \label{thm:main}
  Let $\mathcal{M}_\pi = \{M_1, M_2,... M_l\}$ be a Morse Decomposition for $E \cap \Gamma$ with respect to (\ref{discretemap}). If, for each $M_k\in \mathcal{M}$ there exists $p_{k1},..., p_{kn} > 0$ such that for every $x \in M_k$ there is an $\ell_x \in \mathbb{N}$ satisfying
    \begin{equation}\label{eqxn:persistence_inequality_pulsed}
        \sum^n_{i=1} p_{ki} \sum^{\ell_x}_{j=0} \int_0^\tau f_i(\Phi(\pi(x,j),s)) + \frac{\ln{(1+h_i)}}{\tau} ds > 0
    \end{equation}
    then (\ref{eqxn:pulsed_model}) is permanent.
\end{theorem}

What is valuable about these results in comparison to other persistence results of discrete maps is that we can partition the conditions into the influence of the continual flow and the impact of the pulse disturbance.  
To tease apart this condition, the integral is computing the total growth rate of species $i$ in one period. Provided one can find a weighting over species so that the weighted average of these growth rates for some number of periods is positive, then this condition will be satisfied.  One can observe from this that the impact of the pulse is two-fold: first, it has a direct effect on the growth rate of species $i$ through its instantaneous effect on the population (captured in the second term within the integral) and indirect effect (captured in the first term of the integral) by modifying the trajectories of the all species. 

The central idea behind the proof, which is illustrated in Figure \ref{fig:three_dyn_systems_overview}, is to first embed System (\ref{eqxn:pulsed_model}) into an autonomous continuous system with feedbacks and then exact permanence results of systems of that form.

In particular, we relate system (\ref{eqxn:pulsed_model}) with the autonomous system of equations

\begin{align}\label{auto}
     \frac{dx_i}{dt} &= x_ig_i(x, \theta)\\
    \frac{d\theta}{dt} &=\frac{2\pi}{\tau+1} \label{auto2}
\end{align}
with $x \in \mathbb{R}_+^n, \theta \in S^1 $ (one-dimensional sphere) and 

\begin{align}
    g_i(x, \theta) = \begin{cases}
        f_i(x), \; \theta \in [0,\frac{2\pi\tau}{\tau+1}) \\
        \ln(1+h_i), \; \theta \in [\frac{2\pi \tau}{\tau + 1},2\pi)
    \end{cases}
\end{align}

We let $\phi:\R_+^n \times S^1 \times \R \rightarrow \R_+^n \times S^1$ be the associated dynamical system.  To arrive at this system of equations and explain its relation to the original impulse system, we first relate (\ref{eqxn:pulsed_model}) to a system of nonautonomous but continuous system with the effect of the pulse extended over one unit time, instead of the instantaneous pulse (as was done in \cite{patel_spectral_2024}). Then, we transform the nonautonomous (periodic) system to an equivalent autonomous system (\ref{auto}-\ref{auto2}). A schematic and details on this can be found in Figure \ref{fig:three_dyn_systems_overview} and the Appendix, respectively. 
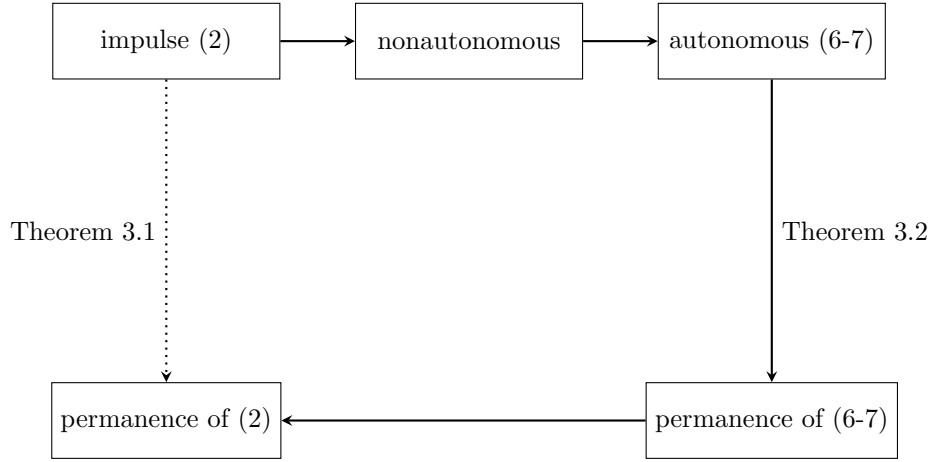
\begin{figure}
   \begin{center}
    \begin{tikzpicture}[node distance=4cm]
    
    \node (box1) [box] {impulse  (\ref{eqxn:pulsed_model})};
    \node (box2) [box, right of=box1] {nonautonomous};
    \node (box3) [box, right of=box2] {autonomous (\ref{auto}-\ref{auto2})};
     \node (box4) [box, below=of box3] {permanence of (\ref{auto}-\ref{auto2})};
     \node (box5) [box, below=of box1] {permanence of (\ref{eqxn:pulsed_model})};
     
    \draw [arrow] (box1) -- (box2);
    \draw [arrow] (box2) -- (box3);
    \draw [arrow] (box3) -- node[right] {Theorem \ref{thm:patel_schreiber_1}}(box4);
    \draw [arrow] (box4) -- (box5);
    \draw [arrow, dotted] (box1) -- node[left] {Theorem \ref{thm:main}}(box5);

    \end{tikzpicture}
    \end{center}
    \caption{Schematic of proof approach  To prove Theorem \ref{thm:main} for focal impulse system (\ref{eqxn:pulsed_model}), we correspond to an autonomous system with feedbacks (\ref{auto}-\ref{auto2}) and apply a prior permanence result (Theorem \ref{thm:patel_schreiber_1}).}
    
    \label{fig:three_dyn_systems_overview}
\end{figure}

System (\ref{auto}-\ref{auto2}) are written to ensure that that the following two conditions hold:
\\

\noindent \textbf{C1:} one revolution around $S^1$ takes $\tau+1$ time units, and

\noindent \textbf{C2:} for any $x\in \R_+^n$ and $m \in \N$, $(\pi(x,m), 0) = \phi((x,0),m(\tau+1))$
\\

Additionally, assumption \textbf{A2} for $\pi$ gives us \\

\noindent \textbf{C3:} there exists a compact set $R \subset \R_+^n \times S^1$ such that $\phi(R) \subset R $ and for any $x\in \R_+^n \times S^1$ there is a $t$ such that $\phi(x,t) \in R$.  \\

As before, for any two sets $R_1$ and $R_2$ satisfying \textbf{C3}, $\omega_\phi(R_1)=\omega_\phi(R_2)$. 

For system (\ref{auto}-\ref{auto2}), we apply previous results from \cite{patel_schreiber_2018} on the permanence of autonomous differential equations with feedbacks. We restate their result, adapted to our context:

\begin{theorem}[Patel and Schreiber 2018]\label{thm:patel_schreiber_1}
    Let $\mathcal{M}_\phi = (M_1, M_2, ..., M_\ell)$ be a Morse decomposition of $\Gamma_\phi \cap (E\times S^1)$ where $\Gamma_\phi = \omega_\phi(R)$.  Then, if for each $M_k$, there is a $\vec{p_k}>0$ such that for all $(x,\theta) \in M_k$ there exists a $T_x$ such that 
    \begin{equation}\label{JMBcondition}
        \sum_{i=0}^n p_{k,i} \int_0^{T_x} g_i(x(s), \theta)ds >0 
    \end{equation}
    then (\ref{auto}-\ref{auto2}) is permanent. 
\end{theorem}

Note that this theorem has a Lipschitz continuity condition to get existence and uniqueness of solutions.  In our equations, we have a piece-wise continuous function.  Our existence and uniqueness follows from examining solutions in a piece-wise manner.

We highlight the strength of this previous result by showing how it can be applied to impulse differential equations of the form (\ref{eqxn:pulsed_model}) through an appropriate embedding.  To apply this result from \cite{patel_schreiber_2018}, we need (I) a correspondence of the Morse decompositions and (II) to show that condition (\ref{eqxn:persistence_inequality_pulsed}) in our result implies condition (\ref{JMBcondition}).  Finally, permanence of (\ref{eqxn:pulsed_model}) follows from permanence of (\ref{auto}-\ref{auto2}).

\subsubsection{Part I: Correspondence of the Morse decomposition}
To write the relationship between $\pi$ and $\phi$, we introduce some notation.  For a set $Q \in \R_+^n \times S^1$, we define the projection onto $S^1$
$$\eta(Q) = \{ s \in S^1 | (x,s) \in Q \text{ for some } x \in \R_+^n \}$$
and for some $s\in S^1$, the cross-section 
$$
Q|_s = \{x\in \R_+^n | (x,s) \in Q \}
$$


We also define a set that maps a point from $\R_+^n$ to the space $\R^n_+ \times S^1$.  That is, for any $x \in \R_+^n$, we define the set


\begin{equation}
    L_x:=  \phi((x,0), [0,\tau+1))
\end{equation}

To relate the Morse decompositions of (\ref{eqxn:pulsed_model}) with that of (\ref{auto}-\ref{auto2}) requires a few lemmas.  In the first lemma, we relate the $\omega$ limit sets between two of the dynamical systems.

\begin{lemma}\label{omega_points}
    For any $z \in \R^n_+$, $$\omega_\phi((z,0)) = \cup_{x\in \omega_\pi(z)} L_x $$
\end{lemma}

\begin{proof}
We first show that 
$$\omega_\phi((z,0)) \supset \cup_{x\in \omega_\pi(z)} L_x $$

Fix $y\in \cup_{x\in \omega_\pi(z)} L_x$.  Then, there is an $x\in \omega_\pi(z)$ such that $y \in \phi((x,0), [0,\tau+1))$, i.e., $y=\phi((x,0), s)$ for some $s\in [0,\tau+1)$. 
Since $x\in \omega_\pi(z)$, there is an increasing sequence $m_j \rightarrow \infty$ such that $\pi(z, m_j) \rightarrow x$.  From this, we define a new increasing sequence $t_j = m_j(\tau+1) + s$.  By continuity with respect to initial conditions, for any $\delta>0$, there is an $\epsilon>0$ such that if $|\tilde{x} - x|<\epsilon$, then $|\phi((x,0),t) - \phi((\tilde{x},0), t)|<\delta$, for any $t\in [0,\tau+1)$.  This, along with condition \textbf{C2}, gives us that $\phi((z,0),t_j) \rightarrow y$, which shows that $y\in \omega_\phi((z,0))$.

Next, we show that 
$$\omega_\phi((z,0)) \subset \cup_{x\in \omega_\pi(z)} L_x $$

Now, fix $y \in \omega_\phi((z,0))$ and let $s=\eta(y)$.  Then, there is an increasing sequence $t_j\rightarrow \infty $ such that $\phi((z,0),t_j) \rightarrow y$.  Let $x = \phi(y, -s)$ so that $y\in L_{x|_0}$.  We will show that $x|_0$ is in $\omega_\pi(z)$. Define the increasing sequence $m_j = \lfloor t_j \rfloor$, where $\lfloor t \rfloor$ is the greatest integer such that $\lfloor t \rfloor (\tau+1) \leq t_j$. Again, by continuity, we have $\phi((z,0), m_j)\rightarrow x$, which by \textbf{C2} implies $\pi(z, m_j)\rightarrow x|_0$.  Hence, $x|_0 \in \omega_\pi(z).$
\end{proof}

The analogous statement for $\alpha-$limits is true:

For any $z\in\R_+^n$, $$\alpha_\phi((z,0)) = \cup_{x\in \alpha_\pi(z)} L_x $$

Additionally, we have a generalization of this lemma for $\omega-$limits of sets:

\begin{lemma}\label{omega_sets}
    For any set $A\subset \R_+^n,$
    $$
    \omega_\phi(A\times \{0\}) = \cup_{x\in \omega_\pi(A)} L_x 
    $$
\end{lemma}

The proof follows similarly, but with choosing a sequence $z_j \in A$ along with $t_j$ to get convergence of $\phi(z_j, t_j)$. 


Next, we show that the global attractor of $\pi$ yields a global attractor for $\phi$.

\begin{lemma}
    If $\Gamma_\pi$ is the global attractor for $\pi$, then  $\tilde{\Gamma}$ defined as

\begin{equation}
    \tilde{\Gamma} = \cup_{x\in \Gamma_\pi} L_x
\end{equation}
is a global attractor for $\phi$, i.e., $\tilde{\Gamma} = \omega_\phi(R).$
\end{lemma}

\begin{proof}
\begin{align*}
    \cup_{x\in \Gamma_\pi} L_x &= \omega_\phi(Q\times \{0\})\\
    & = \omega_\phi (\phi(Q\times \{0\}, [0,\tau+1))\\
    & = \omega_\phi (\cup_{x\in Q} L_x)\\
    & = \omega_\phi(R)
\end{align*}
The first equality follows from Lemma \ref{omega_sets} since $\Gamma_\pi=\omega_\pi(Q)$.
The second follows since for any subset $A\subset \R_+^n\times S^1$, we have $\omega_\phi(A) = \omega_\phi(\phi(A,T))$. 
The third equality is just definitional.  The final equality is since the set $\cup_{x\in Q} L_x)$ satisfies the properties in \textbf{C3} and hence, by the uniqueness of the global attractor, the $\omega-$limits are equal.

\end{proof}

Finally, we show how the above lemmas gives us a relation between a Morse decomposition of $\pi$ to a Morse decomposition for the dynamical system of (\ref{auto}-\ref{auto2})

\begin{lemma}\label{lemma:md_analog}
    If $\mathcal{M}_\pi = (M_1, M_2, ..., M_k)$ is a Morse decomposition for a compact invariant set $\Gamma$ of $\pi$, then $\tilde{\Gamma}$ defined as

\begin{equation}
    \tilde{\Gamma} = \cup_{x\in \Gamma} L_x
\end{equation}
    is a compact invariant set of $\phi$  and 
    $\tilde{\mathcal{M}}= (\tilde{M_1}, \tilde{M_2},...\tilde{M_k})$ 
    where
    $$
    \tilde{M_i} = \cup_{x\in M_i} L_x
    $$
    is a Morse decomposition for $\tilde{\Gamma}$.
\end{lemma}

\begin{proof}
  We first show pairwise disjoint: that $\tilde {M_i} \bigcap \tilde{M_j} = \emptyset$ for all $i \neq j$. This follows easily from the pairwise disjoint property of $\mathcal{M}_\pi$ and the uniqueness of solutions. Similarly, the invariance of each $\tilde{M_i}$ follows from the invariance of $M_i$ and property \textbf{C2}.  Next, we show that each $\tilde{M_i}$ is isolated. Suppose the contrary that $\tilde{M_i}$ is not isolated.  Then, for all $\epsilon_m=\frac{1}{m}$, there exists an invariant set $A_m$ in the neighborhood $U_{\epsilon_m}$ around $\tilde{M_i}$ that is not part of $\tilde{M_i}$.  Then, $A|_0$ is invariant under $\pi$ and not in $M_i$, which is a contradiction to the isolated property of $M_i$. 
Finally, by Lemma \ref{omega_points}, we conclude that for all $z \in \tilde{\Gamma} \setminus \bigcup^k_{i=1} \tilde{M}_i$, there exists $1 \leq i < j \leq k$ such that $\omega_\phi(z) \subset \tilde{M}_i$ and $\alpha_\phi(z) \subset \tilde{M}_j$.

     
    
    




\end{proof}

\subsubsection{Part II: Implications of Conditions}
We are now ready to prove our main permanence result. Having established the correspondence between Morse Decompositions of (\ref{discretemap}) and (\ref{auto},\ref{auto2}), all that remains is to show that condition (\ref{eqxn:persistence_inequality_pulsed})  on (\ref{discretemap}) implies  (\ref{JMBcondition}) holds for (\ref{auto},\ref{auto2}) and, hence, that  (\ref{auto},\ref{auto2}) is permanent. Permanence of (\ref{auto},\ref{auto2}) implies permanence of (\ref{eqxn:pulsed_model}). 

\begin{proof}[Proof of Theorem \ref{thm:main}]
    Let $\mathcal{M}_\pi$ be a Morse Decomposition for $E\cap \Gamma$ with respect to (\ref{discretemap}), and let $\mathcal{M}_\phi$ be the analogous Morse Decomposition of (\ref{auto},\ref{auto2}) guaranteed by Lemma \ref{lemma:md_analog}. 

    Fix an $M_k \in \mathcal{M}_\pi$ and suppose we can find $p_{k1}, \ldots, p_{kn} > 0$ such that for any $z \in M_k$ there is an $l_z > 0$ so that (\ref{eqxn:persistence_inequality_pulsed}) 
    
   Let the function $F: \R^n_+ \times \N \to \R$ be defined
   $$
        F(z,l) = \sum^n_{i=1} p_{ki} \sum^{l_z}_{j=0}\int^{\tau}_0 f_i(\Phi(\pi(z,j), s) + \frac{\ln{(1+h_i)}}{\tau} \;ds.
   $$
  By hypothesis, for each $z \in M_k$ we can find $l_z > 0$ satisfying $F(z,l_z) > 0$.   Hence, for each $z\in M_k$, there exists a $c_z>0$ such that $F(z,l_z) > c_z$.  First, we show that there is a uniform lower bound $c>0$.  By continuity of $F$ with respect to $z$, there exists a neighborhood $V_z$ such that $F(v, l_z) > \frac{c_z}{2}$ for all $v \in V_z$. Note that $\bigcup_{z\in M_k} V_z$ forms an open cover of $M_k$. So by compactness of $M_k$, we can choose $K$ finitely many points $\{z_1, z_2, \ldots, z_K\}$ such that $M_k \subset \bigcup^K_{m=1} V_{z_m}$.

    Let 
    $$
        c = \min {\{\frac{c_{z_m}}{2}\}}_{m=1}^K.
    $$
    Then, for any $z \in M_k$, there exists an $L(z) \in \{l_{z_1},\ldots l_{z_K}\}$ such that $F(z,L(z)) \geq c$.

    Next, we use this uniform bound to show condition (\ref{JMBcondition}). Let $(z_0, \theta_0)$ be any point belonging to the analogous set $\tilde {M_k} \in \mathcal{M}_\phi$. Define $\delta \geq 0$ to be such that  
    $$
    \theta_0 + \frac{2\pi}{\tau + 1}\delta= 2\pi
    $$
    Let $\alpha=\sum^n_{i=1} p_{ki} \int_0^\delta g_i(x(s),\theta(s)) \; ds$. Then, we can find a $N \in \N$ such that
    $$
     Nc > |\alpha|.
    $$
   
   Let $z_1 = \phi((z_0, \theta_0), \delta)|_0.
   $
   Since $\tilde{M}_k$ is invariant, $(z_1,0)$ is contained in $\tilde{M}_k$. Also, $z_1 \in M_k$ by construction of $\tilde{M_k}$. Hence, recall that we can find $L(z_1) \in \N$ satisfying $F(z_1,L(z_1)) > c$. 
Continuing, let $z_2 = \phi((z_0, 0), L(z_1))|_0$ and observe that $z_2 \in M_k$ so we can repeat and find $L(z_2)$ satisfying $F(z_2, L(z_2))>c.$ 
   
    Now, let $L = \sum^N_{m = 1} L(z_m)$ and $\tilde{T}_z = L(\tau + 1) + \delta$. Then, we have
    \begin{equation*}
        \begin{split}
           \sum^n_{i=1} p_{ki} \int_0^{\tilde{T}_z} g_i(x(s),\theta(s)) \; ds  = \sum^n_{i=1} p_{ki} \left[\int_0^\delta g_i(x(s),\theta(s)) \; ds + \int_\delta^{L(\tau + 1)} g_i(x(s),\theta(s)) \; ds \right]\\
            = \sum^n_{i=1} p_{ki} \left[\sum^{L}_{j=0} \int^{(j+1)\tau + j}_{j\tau + j} f_i(\Phi(\pi(z_1, j),s)) \;ds + \int^{(j+1)(\tau + 1)}_{(j+1)\tau + j} \ln{(1+h_i)} \;ds\right] + \alpha\\
            = \sum^n_{i=1} p_{ki} \left[ \sum^{L}_{j=0} \int^{(j+1)\tau + j}_{j\tau + j} f_i(\Phi(\pi(z_1, j),s)) + \frac{\ln{(1+h_i)}}{\tau} \; ds\right] + \alpha\\
            \geq Nc + \alpha \\
            > 0.
        \end{split}
    \end{equation*}
    This gives inequality (\ref{JMBcondition}).  The second and third line rely on \textbf{C2}. In particular, for each $j$, $x(j\tau + j) = \pi(x,j)$, so for $s > \delta$, $x(s) = \Phi(\pi(z_1, j), s)$.
    Thus, by Theorem \ref{thm:patel_schreiber_1}, the embedded autonomous system (\ref{auto},\ref{auto2}) is permanent.

Finally, permanence of (\ref{auto},\ref{auto2}) implies permanence of (\ref{eqxn:pulsed_model}).
  
\end{proof}

\subsection{Robust Permanence}
In biological systems, parameter measurement and general model uncertainty motivates deeper study of the sensitivity of the model properties to perturbation. To this end, mathematical biologists study a stronger form of permanence, known as a robust permanence, in which the permanence of appropriately defined neighboring differential equation systems is also considered.
In our framework, we are interested in perturbations or uncertainties in both the intrinsic continuous population dynamics as well as the timing and impacts of the pulse disturbances.

Here, we refer to perturbations of our model (\ref{eqxn:pulsed_model}) by 

\begin{definition}[($\delta, Q$)-perturbation] Let $Q \subset \R^n_+$ be a compact, forward invariant set with respect to (\ref{eqxn:pulsed_model}). Then $(\tilde{f}, \tilde{\tau}, \tilde{h})$ is a ($\delta,Q$)-perturbation of (\ref{eqxn:pulsed_model}) if
\begin{enumerate}
     \item  For all $i$ and all $x \in \R^n_+$, $|\tilde{f}_i(x)-f_i(x)| < \delta, \; |\tilde{\tau}-\tau| < \delta$, and $|\tilde{h}_i - h_i| < \delta$,
     \item $\tilde{\tau} > 0$ and $\tilde{h}_i > -1$,
     \item $x_i\tilde{f}_i$ is locally Lipschitz continuous, and 
    \item solutions to the perturbed  system $\tilde{x} \in \R^n_+$ defined by
    \begin{equation}\label{eqxn:perturbed_pulsed_model}
    \begin{cases}
        \frac{d\tilde{x}_i}{dt} = \tilde{x}_i(t)\tilde{f}_i(\tilde x), & i = 1, \ldots, n; \; t \neq k\tilde \tau; k \in \Z_+\\
        \tilde{x}_i(k\tilde \tau^+) = (1 + \tilde{h}_i)\tilde{x}_i(\tilde \tau^-)
    \end{cases}
    \end{equation}
    satisfy
       $$
    \{\tilde{x}(s) | s \in \R_+; \; \tilde x(0) \in Q\}  \subset Q
    $$
   and 
   \item for all $\tilde x(0)\in \R^n_+$ there exists a $t \in \R_+$ such that $\tilde{x}(t) \in Q$.
\end{enumerate}

\end{definition}

In words the last two items ensure that the solutions of the perturbed system eventually enter and remain in a compact set. 

\begin{definition}[Robust Permanence]\label{def:robust_permanence} System (\ref{eqxn:pulsed_model}) is called \emph{robustly permanent} if there exists $\delta, \beta > 0$ such that if $(\tilde{f}, \tilde{\tau}, \tilde{h})$ is in the ($\delta, Q$)-perturbation of (\ref{eqxn:pulsed_model}), then for any initial condition with $\tilde{x}_i(t_0) > 0$ and associated solution $\tilde{x}(t)$ of (\ref{eqxn:perturbed_pulsed_model}), we have
    \begin{equation*}
        \frac{1}{\beta} < \liminf_{t \to \infty} \tilde{x}_i(t) < \limsup_{t \to \infty} \tilde{x}_i(t) < \beta
    \end{equation*} for all $i$.
\end{definition}

Our robust permanence result is 

\begin{theorem} \label{thm:robust_permanence}
    If (\ref{eqxn:pulsed_model}) satisfies condition (\ref{eqxn:persistence_inequality_pulsed}) for permanence, then (\ref{eqxn:pulsed_model}) is also robustly permanent. 
\end{theorem}

We put the proof in the appendix. The main idea of the proof relies on Theorem 2 in \cite{patel_schreiber_2018}, which shows that the conditions outlined for permanence of (\ref{auto}-\ref{auto2}) also give robust permanence (under an analogous definition as above).  

\section{Examples}
There are many ecological management problems where human interventions can be modeled as pulsed systems. For instance, impulsive population models are used to plan the reintroduction of endangered species \cite{xu_reintroduction_2016}, sustainable impulsive (as opposed to continuous) harvesting \cite{xu_harvesting_2005,yoshioka_unpulsed_ode_control_2024}, and biological control \cite{sentis_biocontrol_modeling_2022,mailleret_pulsed_biocontrol_2009}. 

Here we illustrate the potential applications of our permanence results of a broad class of impulse differential equations through two examples. In both, we use our main theorem to identify key control parameter ranges that lead to robust permanence.  Then, we use numerical solutions to illuminate finer qualitative system behaviors when permanence fails. In the first example, we explore a cancer system of competing chemo-resistant and chemo-sensitive cells similar to the general model studied by Jin et al in \cite{jin_persistence_2005} and give an alternate proof of their permanence conditions. In the second, we examine a predator-prey model in the context of integrated pest management. These scenarios highlight cases where total elimination of harmful biological population is not feasible and system permanence is desirable to enable biological control by a second population.

\subsection{Maintenance Chemotherapy in a Two Species Competition Model}
Despite the progress of medical science and modern cancer treatments, there remain many cancer diagnoses where complete remission is not possible. In such cases, drug and radiation therapy can sometimes be used to manage cancer as a chronic illness, with the goal of stabilizing disease progression and minimizing patient side effects \cite{lee_maintenance_chemo_2014}. The success of this approach is limited by available treatments, patient tolerance, and, critically, the development of chemo-resistance in cancer cells \cite{strobl_cost_of_resistance_2021, enriquez_adaptivetherapy_2015}. Modelling of cancer cell system dynamics can be used to optimize the effectiveness of chemotherapeutic control methods by reducing selective pressure for chemo-resistant cells \cite{butner_mathematical_oncology_2022,kozlowska_maintenance_chem_2024,barbolosi_computation_oncology_2016}. 

Here, we consider a cancer system of chemo-sensitive ($x_s$) and chemo-resistant ($x_r$) type cells treated $\tau$-periodically with a fixed dose of chemotherapy that works with efficacy $h_s$ against sensitive-type cells and $h_r$ against resistant-type cells. We will assume that the per capita growth rate of the sensitive-type and resistant-type cells in the untreated cancer system are dependent on the intrinsic reproductive rates ($a_s$ and $a_r$), the rate of intraspecies competition ($b_{ss}$ and $b_{rr}$), and the rate of interspecies competition ($b_{sr}$ and $b_{rs}$), and that dynamics can be described by the following 2-dimensional Lotka Volterra competition model.
\begin{equation}\label{eqxn:lv_cancer_model}
\begin{split}    
\frac{d x_s}{d t} = x_s(a_s - x_sb_{ss} - x_rb_{sr}), \; t \neq k\tau, k \in \N \\
\frac{dx_r}{dt} = x_r(a_r - x_sb_{rs} - x_rb_{rr}),  \; t = k\tau\\
x_i(k\tau)= (1 + h_i)x_i(k\tau^-). 
\end{split}
\end{equation}

Jin et al \cite{jin_persistence_2005} studied conditions for system permanence in a model similar to system (\ref{eqxn:lv_cancer_model}) using comparison theorems. We will demonstrate how their permanence result can be reproduced using Theorem \ref{thm:main} and elaborate on some distinctions and advantages provided by our main theorem in the context of cancer treatment models.

We begin by discussing cancer cell dynamics in a single cell type system, that is (\ref{eqxn:pulsed_model}) with $n=1$ so that
    \begin{equation}\label{eqxn:lv_1d}
    \begin{split}
        \frac{dx}{dt} = x(a - b x), \; t \neq k\tau; k \in \Z_+ \\
        x(k\tau) = (1+h)x(\tau^-)
    \end{split}
    \end{equation}
where $a, b > 0$. We refer to the function $a -bx$ as $f(x)$ the per capita growth rate. Jin et al \cite{jin_persistence_2005} showed that the existence of a periodic solution to system (\ref{eqxn:lv_1d}) is dictated as follows.

\begin{lemma}[Jin et al.; Lemma 2.4]\label{lemma:jin_1d_periodic_soln}
    Consider (\ref{eqxn:lv_1d}). If
    \begin{equation*}
        \ln{(1+h) + \tau a > 0}
    \end{equation*}
    then (\ref{eqxn:lv_1d}) has a unique $\tau$-periodic solution $x^*(t)$, and $x^*(t)$ is globally, asymptotically stable in the sense that $\lim_{t \to \infty}|x(t) - x^*(t)| = 0$, where $x(t)$ is any solution of system (\ref{eqxn:lv_1d}) with positive initial value $x(0) > 0$. Conversely, if
    \begin{equation*}
        \ln{(1+h) + \tau a < 0}
    \end{equation*}
    then $\lim_{t \to \infty}|x(t)| = 0$ where $x(t)$ is any solution of system (\ref{eqxn:lv_1d}) with positive initial value $x(0) > 0$.
\end{lemma}

Later, we will see that this permanence condition for (\ref{eqxn:lv_1d}) is a necessary assumption on $x_s$ and $x_r$ for permanence of the two-dimensional system. Even without this foreknowledge, characterizing the permanence of system (\ref{eqxn:lv_1d}) is useful for elucidating system dynamics when two-species permanence fails. For now, a quick consequence of Lemma \ref{lemma:jin_1d_periodic_soln} is that population growth in the single-species system is bounded. The lemma below summarizes this observation.

\begin{lemma}\label{lemma:periodic_integral_of_lv_1d}
    Consider (\ref{eqxn:lv_1d}). If there exists a nontrivial $\tau$-periodic solution $x^*$ for (\ref{eqxn:lv_1d}) such that $x^*(t + \tau) = x^*(t)$, then the following statements are true for all solutions $x(t)$ with initial conditions in $x^*$.
    \begin{equation}\label{eqxn:periodic_integral_of_f}
        \int_0^\tau f(x) \; dt = -\ln{(1 + h)}
    \end{equation}
    \begin{equation}\label{eqxn:periodic_integral_of_x}
        \int_0^\tau x(t) \; dt = \frac{a\tau  + \ln{(1+h)}}{b} 
    \end{equation}
           
\end{lemma}

\begin{proof}
    We first prove (\ref{eqxn:periodic_integral_of_f}).
    
    By rearrangement of (\ref{eqxn:lv_1d}) for $f$ we have
    \begin{equation*}
        f(\vec{x}(t)) = \frac{\dot{x}(t)}{x(t)} = \frac{d}{dt}\ln(x(t))
    \end{equation*}
    Thus
    \begin{equation*}
        \int_0^\tau f(x(t)) \; dt = \ln(x(\tau^-)) - \ln(x(0))
    \end{equation*}
    The existence of $x^*$ implies $x(0) = x(\tau)$. From our model definition $x(\tau)=x(\tau^-)(1+h)$, so we have
    \begin{equation*}
        \begin{split}
            \ln(x(\tau^-)) - \ln(x(0)) = \ln(x(\tau^-)) - \ln(x(\tau^-)(1+h))\\
            = \ln{\left(\frac{x(\tau^-)}{x(\tau^-)(1+h)}\right)} \\
            = - \ln{(1 + h)}
        \end{split}
    \end{equation*}
    and (\ref{eqxn:periodic_integral_of_f}) holds.

    To prove (\ref{eqxn:periodic_integral_of_x}), we first apply (\ref{eqxn:periodic_integral_of_f}). This gives
    \begin{equation*}
        \begin{split}
           \int_0^\tau a - bx(t) \; dt = \int_0^\tau f(x(t)) \; dt = 
             -\ln{(1+h)}\\
        \end{split}
    \end{equation*}
    Rearranging yields
    \begin{equation*}
        \begin{split}
            \int_0^\tau x(t) \; dt = \frac{a\tau  + \ln{(1+h)}}{b} 
        \end{split}
    \end{equation*}
    
\end{proof}

Returning to the full system (\ref{eqxn:lv_cancer_model}), we are now ready to prove and strengthen the main result of \cite{jin_persistence_2005}.

\begin{theorem}
    Consider (\ref{eqxn:lv_cancer_model}). Assume
    \begin{equation}\label{eqxn:lv_cancer_model_permanence_req_1}
        \tau a_s + \ln{(1 + h_s)} > 0 , \; \tau a_r + \ln{(1 + h_r)} > 0
    \end{equation}
    If 
    \begin{equation}\label{eqxn:lv_cancer_model_permanence_req_2}
        \begin{split}
            a_s\tau + \ln{|1+h_s|} > \left(a_r\tau + \ln{|1 + h_r|}\right)\frac{b_{sr}}{b_{rr}}\\
            a_r\tau + \ln{|1+h_r|} > \left(a_s\tau + \ln{|1 + h_s|}\right)\frac{b_{rs}}{b_{ss}}
        \end{split}
    \end{equation}
    then the system is robustly permanent.
\end{theorem}
\begin{proof}
    Let $\pi$ be the discrete time map for (\ref{eqxn:lv_cancer_model}) defined as in (\ref{discretemap}).
    
    We first construct an appropriate Morse decomposition $\mathcal{M}_\pi$. The extinction set $E$ of (\ref{eqxn:lv_cancer_model}) is the set $\{(x_r,x_s)| x_rx_s=0\}$. By (\ref{eqxn:lv_cancer_model_permanence_req_1}) and Lemma \ref{lemma:jin_1d_periodic_soln}, there are unique $\tau$-periodic solutions $x_r^*(t)$ and $x_s^*(t)$ in the sets $\{(0, x_r)|x_r>0\}$ and $\{((x_s,0)|x_s>0\}$, respectively.  From these periodic solutions, we define the sets $M_1$ and $M_2$ in $E$. Then, if $\Gamma$ is the global attractor on (\ref{eqxn:lv_cancer_model}), a valid Morse Decomposition for $E\cap\Gamma$ is  
    $$
    \mathcal{M}_\pi = \{ M_1, M_2, M_3 \}
    $$ with
    where $M_3=\{(0.0)\}$.
    Now, for each $k \in \{1, 2, 3\}$ and $M_k \in \mathcal{M}$ we consider the existence of $p_{k1},p_{k2} >0$ satisfying condition (\ref{eqxn:persistence_inequality_pulsed}).

    \textit{($k = 1$)} 
    Let  $z \in M_1$, so $z = (0,x_s^*(\tau))$. Recall $x_s^*(t)$ is a $\tau$-periodic solution of (\ref{eqxn:lv_cancer_model}) so for all $s \in \R$,
    $$
    \Phi(\pi(z,0),s) = \Phi(\pi(z,j),s), \; j\in \N
    $$ By Lemma \ref{lemma:periodic_integral_of_lv_1d}.\ref{eqxn:periodic_integral_of_f} we have
    $$
    \int_0^\tau f_s(\Phi (\pi(z,j),s)) + \frac{\ln{(1+h_s)}}{\tau} \; dt = -\ln{(1+h_s)} + \ln{(1+h_s)} = 0
    $$
    Likewise, expanding $f_r$ and applying Lemma \ref{lemma:periodic_integral_of_lv_1d}.\ref{eqxn:periodic_integral_of_x} we obtain
    $$
    \int_0^\tau a_r + b_{rs}x_s(t) +b_{rr}x_r(t) \; dt = a_r\tau + b_{rs}\left(-\frac{a_s\tau + \ln{(1+h_s)}}{b_{ss}}\right) + \ln{(1 + h_r)}
    $$ So combined we see that the following expressions are equivalent
    \begin{equation*}
        \begin{split}
            p_{1s} \int_0^\tau f_s(\Phi (\pi(z,j),s)) + \frac{\ln{(1+h_s)}}{\tau} \; ds + p_{1r} \int_0^\tau f_r(\Phi (\pi(z,j),s)) + \frac{\ln{(1+h_r)}}{\tau} \; ds=\\
            p_{1r} \left[a_r\tau + b_{rs}\left(-\frac{a_s\tau + \ln{(1+h_s)}}{b_{ss}}\right) + \ln{(1 + h_r)} \right]\\
        \end{split}
    \end{equation*}
    Therefore, since $a_r\tau + \ln{(1 + h_r)} > \left(a_s\tau + \ln{|1 + h_s|}\right)\frac{b_{rs}}{b_{ss}}$, for $p_{1s}, p_{1r} = 1$, inequality (\ref{eqxn:persistence_inequality_pulsed}) is satisfied for all $z \in M_1$. \\
    
    \textit{($k = 2$)}  
    By symmetry, we may interchange $x_s$ and $x_r$ in the $k = 1$ case above to see that $a_s\tau + \ln{|1+h_s|} > \left(a_r\tau + \ln{|1 + h_r|}\right)\frac{b_{sr}}{b_{rr}}$ implies
    \begin{equation*}
        p_{2s}\int_0^{\tau} f_s(\Phi(\pi(z,j),s)) + \frac{\ln{(1+h_s)}}{\tau}\; ds +  p_{2r} \int_0^{\tau} f_r(\Phi(\pi(z,j),s)) + \frac{\ln{(1+h_r)}}{\tau}\; ds > 0
    \end{equation*}
    is true for $p_{2s}, p_{2r} = 1$ and any $z \in M_2$.

    \textit{($k = 3$)}
    Expanding $f_s(z), f_r(z)$ with $z = (0,0)$ gives $f_i(z) = a_i$. Hence, we have
    \begin{equation*}
    \begin{split}
        p_{3s} \int_0^{\tau} f_s(\Phi(\pi(z,j), s)) + \frac{\ln{(1+h_s)}}{\tau}\; ds +  p_{3r} \int_0^{\tau} f_r(\Phi(\pi(z,j), s)) + \frac{\ln{(1+h_r)}}{\tau}\; ds \\
        =
        p_{3s} (a_s \tau + \ln{(1+h_s)}) + p_{3r}( a_r \tau + \ln{(1+h_r)})
    \end{split}
    \end{equation*}
    So by (\ref{eqxn:lv_cancer_model_permanence_req_1}) for $p_{3s}, p_{3r} = 1$, (\ref{eqxn:persistence_inequality_pulsed}) holds for $z = (0,0)$. 
    
    Therefore for each $M_k \in \mathcal{M}$, there are $p_{ks}, p_{kr} > 0$ satisfying (\ref{eqxn:persistence_inequality_pulsed}), and by Theorem \ref{thm:main} and Theorem \ref{thm:robust_permanence}, (\ref{eqxn:lv_cancer_model}) is robustly permanent.
\end{proof}

We find it instructive to compare this result with permanence conditions in the classic Lotka-Volterra competition model. Recall, that for two species the following two conditions give permanence:
$$
a_s > a_r\frac{b_{sr}}{b_{rr}}, \; a_r > a_s\frac{b_{rs}}{b_{ss}}
$$
These inequalities can also be found by taking $h_s = h_r = 0$ and $\tau = 1$ in condition (\ref{eqxn:lv_cancer_model_permanence_req_2}). Thus, intuitively we can interpret the expressions $a_s\tau +\ln{(1+h_s)}$ and $a_r\tau + \ln{(1+h_r)}$ as the adjusted intrinsic reproductive rate in the pulsed system (\ref{eqxn:lv_cancer_model}).

Moving to biological interpretation, we now illustrate how the results of Theorem \ref{thm:main} may be used to predict treatment outcomes in cancer maintenance chemotherapy.

Assuming that both the chemo-resistant and chemo-sensitive cells are naturally present in some proportion prior to beginning treatment ($x_r(0), x_s(0) > 0$), we examine three cases arising in cancer systems with competing resistant and sensitive cells:

\noindent \textbf{Case 1:} resistant-type cells are identical to and coexist with sensitive cells in the untreated system \\
\noindent \textbf{Case 2:} resistant-type cells are excluded in the untreated system \\
\noindent \textbf{Case 3:} resistant-type cells exclude sensitive-type cells in the untreated system\\ 

\noindent To reflect the differences in chemo-sensitivity between $x_s$ and $x_r$, we let $h_0 \leq 0$ be the baseline intensity of the chemo dose applied with $h_s = h_0$ and $h_r = -|h_0|^2$. 

Starting with Case 1, consider the dynamics of a cancer system where the resistant phenotype has no effect on system dynamics in the absence of chemo treatment, that is $a_s=a_r$, $b_{ss} = b_{rr}$, and $b_{sr}=b_{rs}$ as in Figure \ref{fig:cancer_unpulsed_coexistence}A. Plotting conditions (\ref{eqxn:lv_cancer_model_permanence_req_1}) and (\ref{eqxn:lv_cancer_model_permanence_req_2}), we obtain $\tau$ and $h_0$ parameter ranges leading to permanence of only chemo-resistant cells and system permanence, respectively (note that mutual extinction of chemo-resistant and chemo-sensitive cells occurs when neither condition (\ref{eqxn:lv_cancer_model_permanence_req_1}) nor (\ref{eqxn:lv_cancer_model_permanence_req_2}) is met).  For sufficiently strong or frequent interventions, all cancer cells are eliminated (Figure \ref{fig:cancer_unpulsed_coexistence}B). However, this may not always be possible due to resource limitations or patient-tolerance.
Another part of parameter space reflects treatments that yield outcomes in which resistant cells outcompete sensitive cells. 
This scenario is undesirable from a medical perspective as it corresponds to a system where chemotherapy is no longer an effective form of cancer control.
Hence, there is motivation for adopting a more moderate treatment plan, where coexistence of sensitive and resistant type cells is expected and the goal is to reduce overall cancer load.

\begin{figure}[h]
    \centering
    \includegraphics[width=1\linewidth]{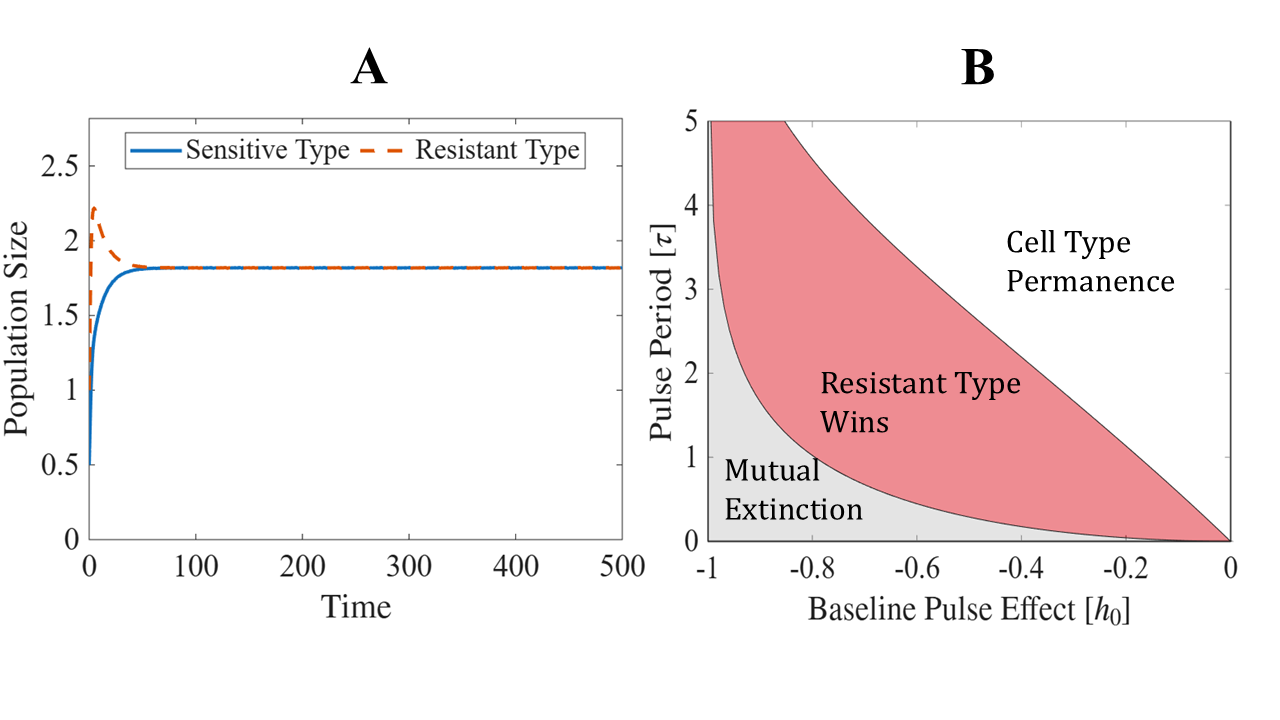}
    \caption{Pulsed treatment dynamics for a cancer system that is permanent in the absence of a pulse. In (A), the numerical solution for a representative permanent, untreated system with $a_s = a_r = 1$, $b_{ss} = b_{rr} = 0.3$, $b_{rs}=b_{sr} = 0.25$, and initial conditions $x_s(0) = 0.5$ and $x_r(0) = 1$. In (B), a bifurcation plot with the same choice of parameterization varying the baseline, pulsed, treatment intensity and pulse period, where the pulse effect on chemo-resistant type and chemo-sensitive type cells is $h_r = -h_0^{2}$ and $h_s = h_0$, respectively.}
    \label{fig:cancer_unpulsed_coexistence}
\end{figure}

Next, we examine Case 2 where success of resistant-type cancer cells only occurs in treatment conditions, that is $\lim_{t\to \infty} x_r(t) = 0$ when $h_0=0$ as in Figure \ref{fig:cancer_unpulsed_resistant_disadvantage}A. Such an outcome can be achieved by various parameter choices. For simplicity, we assume intraspecies competition is uniform across all cell types ($b_{ss} = b_{rr} = 0.3$) and divide our analysis into two subcases:\\

\noindent\textbf{Case 2A:} Sensitive type cells place larger interspecies, competitive pressure on resistant cells ($b_{rs} = 0.4, b_{sr} = 0.25$) but both cell types have similar reproductive rates ($a_s = a_r =1$).\\
\noindent\textbf{Case 2B:} Chemo-resistant cells grow slower in the absence of pulsed therapy ($a_r = 0.5 , a_s = 1$) but cell types are otherwise identical ($b_{rs} = b_{sr} = 0.25$).\\

Both Case 2A and 2B occur in real cancer systems, where chemo-resistance can be associated with increased dependence on certain cell nutrients, possibly increasing intra- or inter-cellular competition (Case 2A) \cite{broxterman_nutritional_cancer_1988,strobl_cost_of_resistance_2021,gatenby_cost_of_resistance_2009}, or delayed doubling times in resistant cells  (Case 2B) \cite{replogle_delayed_growth_cancer_2020}.
Fixing our choice in unpulsed system parameters, and plotting inequalities (\ref{eqxn:lv_cancer_model_permanence_req_1},\ref{eqxn:lv_cancer_model_permanence_req_2}) in the $\tau-h_0$ phase space reveals very different cancer system dynamics between subcases (Figures \ref{fig:cancer_unpulsed_resistant_disadvantage}B \ref{fig:cancer_unpulsed_resistant_disadvantage}C). In Case 2A, Figure \ref{fig:cancer_unpulsed_resistant_disadvantage}B, for $(\tau, h_0) \in [0,5] \times[-1,0]$ there is a higher risk of selecting for resistant-type cells by an improper combination of dosage and treatment period. By contrast, restricting $(\tau, h_0)$ to the same ranges in Case 2B,  Figure \ref{fig:cancer_unpulsed_resistant_disadvantage}C suggests that pulsed therapy can be used to exclude the resistant species in most instances, with coexistence of both cell types or selection for the resistant-type possible in only small regions of the $\tau$-$h_0$ plane. 
Our main result illustrates the delicacy of system permanence and potential value-added to patient care by differentiating between qualitatively similar unpulsed cancer system dynamics, like Case 2A and 2B.

\begin{figure}
    \centering
    \includegraphics[width=1\linewidth]{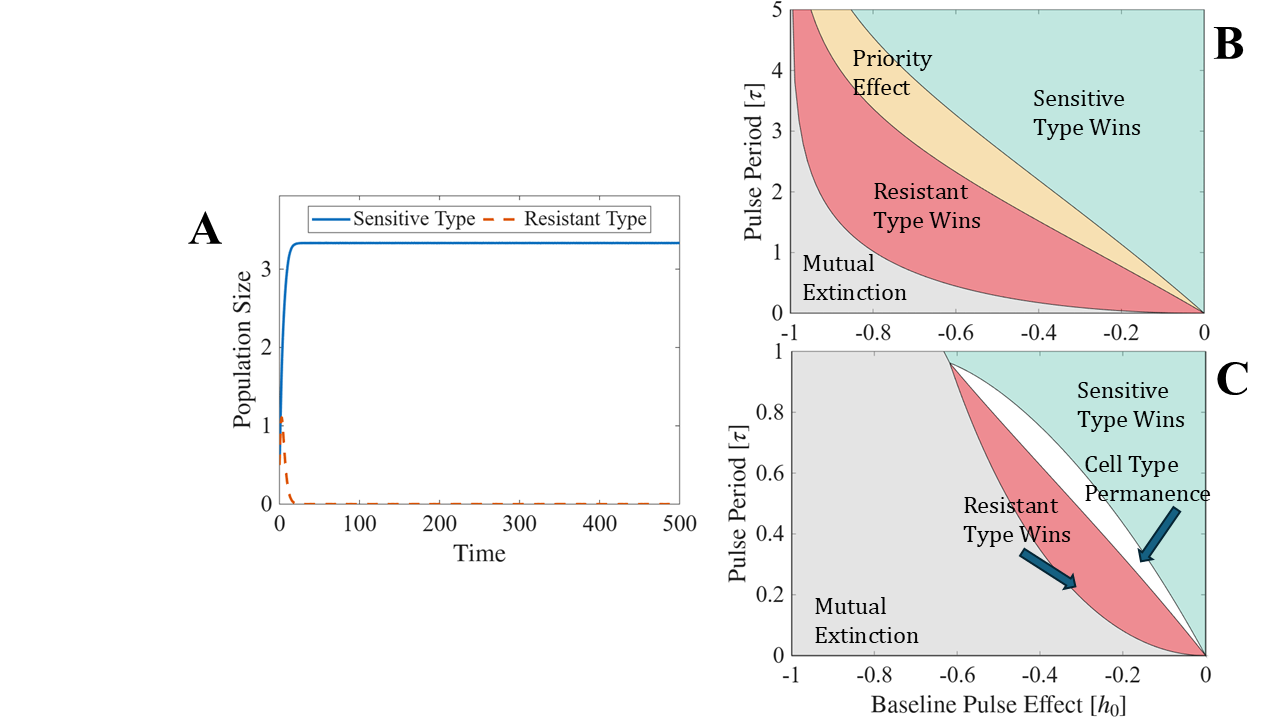}
    \caption{Pulsed treatment dynamics for a cancer system where chemo-resistant type cells lose their competitive advantage in the absence of treatment. In (A), the numerical solution for a representative, untreated system with $a_s = a_r = 1$, $b_{ss} = b_{rr} = 0.3$, $b_{sr}=0.25$,$=b_{rs} = 0.4$, and initial conditions $x_s(0) = 0.5$ and $x_r(0) = 0.5$. In (B), a bifurcation plot varying treatment intensity and pulse periods for the same system where sensitive type cells outcompete chemo-resistant cells in untreated conditions. The region labeled priority effect denotes the parameter space where either the chemo-resistant or chemo-sensitive type population can be excluded depending on initial condition. In (C), a similar bifurcation plot this time where sensitive-type cells have a growth advantage in the untreated system, with $a_r = 0.5$ and $b_{sr}=b_{rs} = 0.25$.
    }
    \label{fig:cancer_unpulsed_resistant_disadvantage}
\end{figure}

Finally, we turn to Case 3, a scenario where resistant-type cells exclude sensitive-type cells in the unpulsed system. An example of a numerical solution exhibiting such behavior is shown in Figure \ref{fig:cancer_unpulsed_resistant_advantage}A. Again, we restrict our focus to two subcases:\\

\noindent\textbf{Case 3A:} Resistant type cells place larger interspecies, competitive pressure on sensitive cells ($b_{sr} = 0.4, b_{rs} = 0.25$) but both cell types have similar reproductive rates ($a_s = a_r =1$).\\
\noindent\textbf{Case 3B:} Chemo-resistant cells grow faster in the absence of pulsed therapy ($a_r = 1 , a_s = 0.5$) but cell types are otherwise identical ($b_{rs} = b_{sr} = 0.25$).\\

As before, there are real chemo-resistance pathways in cancer cells that might be modeled generically by subcases 3A and 3B \cite{zheng_cancer_resistance_overview_2017}. 
Constructing $\tau-h_0$ bifurcation diagrams as in previous cases, results in similar plots in both Case 3A and 3B--- treatment outcomes are limited to mutual extinction and exclusion of the sensitive species (Figure \ref{fig:cancer_unpulsed_resistant_advantage}B and \ref{fig:cancer_unpulsed_resistant_advantage}C).
The information generated by applying our permanence result could still be used to guide treatment decisions and help patients understand their best options.

\begin{figure}
    \centering
    \includegraphics[width=1\linewidth]{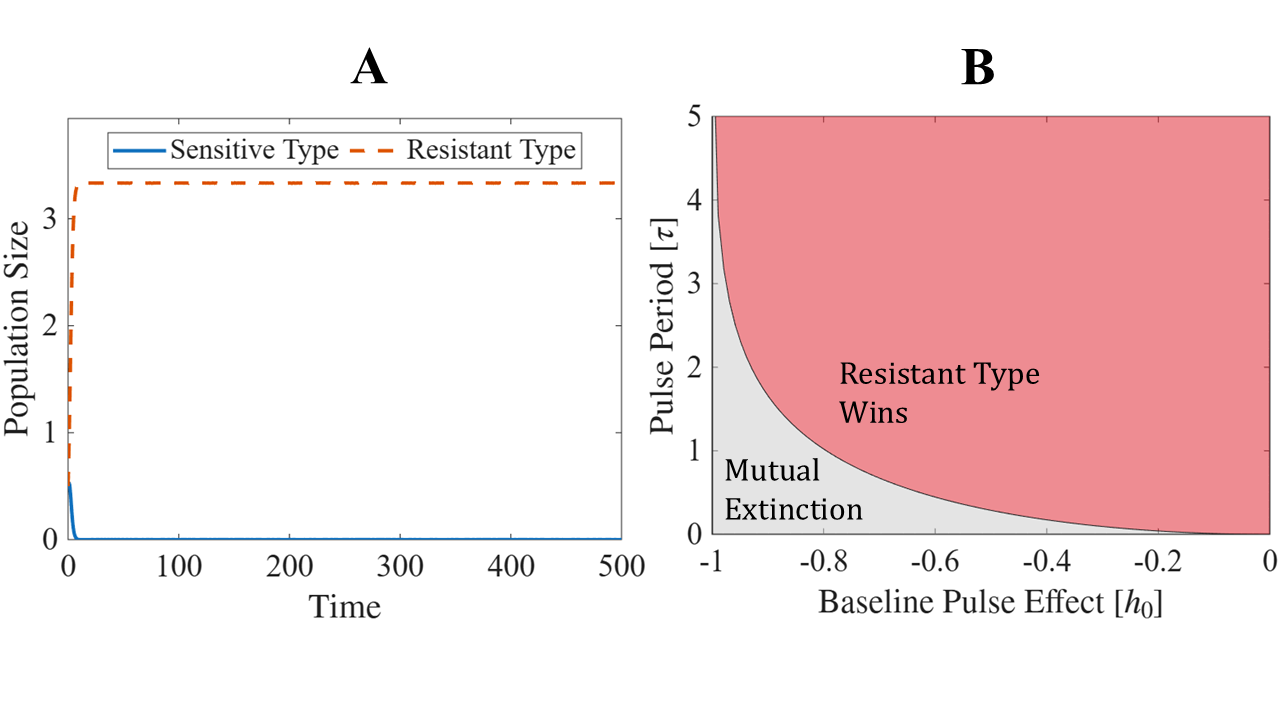}
    \caption{Pulsed treatment dynamics for a cancer system where chemo-resistant type cells have a fitness advantage and exclude sensitive type cells in the absence of a pulse. In (A), the numerical solution for a representative, untreated system with $a_s = a_r = 1$, $b_{ss} = b_{rr} = 0.3$, $b_{sr}=0.4$,$=b_{rs} = 0.25$, and initial conditions $x_s(0) = 0.5$ and $x_r(0) = 0.5$. In (B), a bifurcation plot using the same parameters, now varying the baseline, pulsed, treatment intensity and pulse period, where the pulse effect on chemo-resistant type and chemo-sensitive type cells is $h_r = -h_0^{2}$ and $h_s = h_0$, respectively.}
    \label{fig:cancer_unpulsed_resistant_advantage}
\end{figure}

Overall, we note that an advantage of our main theorem is that it can be applied to a broader set of cancer growth models. Moreover, in all cases, the robust permanence of system (\ref{eqxn:lv_cancer_model}) following from Theorem \ref{thm:robust_permanence} is useful from a treatment standpoint because there is inherent uncertainty when parameterizing prediction models with real patient data. 

\subsection{Parasitoid Control in a Two Species Predator-Prey Model}
{\it Drosophila suzukii}, or the Spotted Winged Drosophila (SWD), is an agricultural pest invasive to Europe and North America \cite{wang_swd_review_2026}. Its partiality to laying eggs in fresh, soft-skinned berries and stone fruit is the cause of serious economic damage in the farming industry \cite{wang_swd_review_2026}. Pesticides are one of the most common methods used to control SWD, but this strategy's effectiveness has been limited by regulations on where and how often farmers may apply treatment \cite{lee_swd_pesticide_2015, haviland_mrl_swd_2012,goodhue_swd_crop_rejection_2011}. On top of this, there are arguments in favor of reducing pesticide usage for practical reasons including resistance development and adverse ecosystem effects \cite{gress_swd_resistance_2019}.

To this end, pesticide treatment augmented with a parasitoid control has been proposed as a viable alternative to pesticide treatment alone \cite{wang_parasitoid_2020,wang_swd_review_2026}. A parasitoid is an insect that uses its host's body to develop through its own life-stages and ultimately destroys its host. Consequently, it may act as a biological aid to control pests. For SWD, wasp parasitoids native to the fly's region-of-origin (Japan) have been studied as a candidate biocontrol agent \cite{wang_swd_review_2026}. In addition to reducing the amount of pesticide treatment necessary to control pests in crops, it has been suggested that these types of plans may reduce resistance development by attacking resistant populations \cite{lee2019biological,liu2014natural, rossistacconi2019augmentative}.

Impulse control models of predator-prey systems have been studied in some contexts \cite{heimpel_biocontrol_overview_2018, xie_predator_prey_example_2017,wang_predator_prey_example_2008}. In the context of SWD--parasitoid dynamics, timing for parasitoid release has been a focus, with possibly multiple impulsive releases \cite{pfab_optimized_2018, mailleret_pulsed_biocontrol_2009}. Here we study a model that focuses on an integrated pest-management strategy combining parasitoid and pesticide control. We focus on simple Lotka-Volterra predator-prey dynamics, 

Our system is 
    \begin{align}
         &\frac{\rd x}{\rd t} = x\left(-a_x + b_{yx}y\right), & t \neq k\tau, k \in \N \label{sys:LVPredatorPreya}\\
&\frac{\rd y}{\rd t} = y\left(a_y - b_{xy}x - b_{yy}y\right),  & t \neq k\tau\\
&x(k\tau^+)= (1 + h_x)x(k\tau^-), &y(k\tau^+)= (1 + h_y)y(k\tau^-)\label{sys:LVPredatorPrey}
    \end{align}
where $x$ is the population size of the parasitoid (predator), and $y$ the SWD (prey). The parameters $b_{yx},b_{xy} > 0$ represent the predation coefficient for the predator and prey, respectively. The parameter $b_{yy}>0$ denotes the intraspecies competition coefficient for the prey. We consider a case where there is no impulsive release of the parasitoid, and the pesticide effects both the populations negatively, i.e. $1 + h_i < 0.$ 

We begin by applying Theorem \ref{thm:main}. 
\begin{theorem}[Permanence for the Lotka-Volterra Predator-Prey System]
When $1 + h_x < 1$ and $\frac{a_x}{b_{yx}} > \frac{a_y}{b_{yy}},$ System (\ref{sys:LVPredatorPreya}- \ref{sys:LVPredatorPrey}) is permanent if 
$$
a_y \tau + \ln(1 + h_y ) > 0
$$
and
$$-a_x \tau + \ln(1 + h_x) > -\frac{b_{yx}}{b_{yy}}\left[a_y \tau + \ln(1 + h_y )\right].$$

\end{theorem}

\begin{proof}
We derive these conditions from our main result. Let $E := \{(x,y): xy =0\}$ and let $\Gamma$ be the global attractor of (\ref{sys:LVPredatorPreya}- \ref{sys:LVPredatorPrey}). When $1 + h_x < 1,$ then the Morse decomposition of $E \cap \Gamma$ for system (\ref{sys:LVPredatorPreya}- \ref{sys:LVPredatorPrey}) is 
$$
\mathcal{M_\pi} = \{M_1, M_2 \}
$$
where 
$$
M_1 =\{(0, y^*(\tau))\}, M_2 = \{(0,0)\}
$$
and $y^*(t)$ is the unique, globally asymptotic, periodic solution of the system in the absence of the predator (i.e., $x = 0.$). Recall by Lemma \ref{lemma:jin_1d_periodic_soln}, that such a solution exists when $a_y\tau + \ln(1 + h_y) > 0$. 

For each $M_k$ consider the existence of $p_{kx}, p_{ky} > 0$ satisfying inequality (\ref{eqxn:persistence_inequality_pulsed}) for system (\ref{sys:LVPredatorPreya}- \ref{sys:LVPredatorPrey}).

\textit{($k=1$)}
By Lemma \ref{lemma:periodic_integral_of_lv_1d}.\ref{eqxn:periodic_integral_of_f} on for initial condition $(0,y^*(\tau))$ 
$$ \int_0^\tau f_y(z(s)) + \frac{\ln(1 + h_y)}{\tau} \ \;d s = 0.$$ 

Additionally since, by definition, we have
$$\int_0^\tau f_y(z(s)) \ \; d s = \int_0^\tau a_y - b_{yy}y \ \;d s,$$
on $M_1$, we find that
\begin{equation}
\int_0^\tau y(t) \ \rd t = \frac{a_y\tau + \ln(1 + h_y)}{b_{yy}}.
\end{equation}
Thus, 
\begin{equation}\label{eqxn:predator_prey_IntxwonPhiw}
\begin{split}
\int_0^\tau f_x(z(s)) \ \;d s &= \int_0^\tau -a_x + b_{yx} y \ \; d s \\
&= -a_x\tau + \frac{b_{yx}}{b_{yy}}(a_y\tau + \ln(1 + h_y)). 
\end{split}
\end{equation}
Now applying our main result and letting $p_{1x} = p_{1y} = \frac{1}{2}$, we see that 
\begin{equation} 
\begin{split}
\sum^2_{i=1} p_{ki} \int_0^{\tau} f_i(z) + \frac{\ln{(1+h_i)}}{\tau} \; ds > 0 \iff \int_0^{\tau} f_x(z(s)) + \frac{\ln(1 + h_x)}{\tau} \ \; d s > 0.
\end{split} 
\end{equation}
Evaluating the integral on the right-hand-side of the equivalence and using Equation (\ref{eqxn:predator_prey_IntxwonPhiw}) we find the condition
\begin{equation} \label{eqxn:predator_prey_full_perm_condition}
-a_x \tau + \ln(1 + h_x) + \frac{b_{yx}}{b_{yy}}\left[a_y \tau + \ln(1 + h_y )\right] > 0.
\end{equation}

\textit{($k = 2$)}
Choose $p_{2x} = 1$ and $p_{2y} = \frac{2|-a_x\tau + \ln{(1+h_x)}|}{a_y\tau + \ln{(1+h_y)}}$. Then
$$
\sum^2_{i=1} p_{ki} \int_0^\tau f_i(z(s)) + \frac{\ln(1+h_i)}{\tau} \; ds =
-a_x\tau + \ln{(1+h_x)} + 2|-a_x\tau + \ln{(1+h_x)}|
$$
which is greater than zero, so condition (\ref{eqxn:persistence_inequality_pulsed}) is satisfied. 

This concludes the proof. 
\end{proof}

We can understand this permanence condition by relating back to the untreated predator-prey system. That is, the differential system
\begin{equation}
\begin{split}
\frac{\rd x}{\rd t} = x\left(-a_x + b_{yx}y\right)\\
\frac{\rd y}{\rd t} = y\left(a_y - b_{xy}x - b_{yy}y\right)
\end{split}
\end{equation}.

Phase-plane analysis shows that if 
$$\frac{a_y}{b_{yy}} > \frac{a_x}{b_{yx}},$$
then the system without a pulse has three equilibria, the trivial extinction state, a prey-only equilibrium, and a coexistence state. Conversely, if 
$$\frac{a_y}{b_{yy}} < \frac{a_x}{b_{yx}},$$
then the coexistence state vanishes, and the only two equilibria are the prey-only equilibrium and the trivial extinction state. 

We observe from permanence conditions (\ref{eqxn:predator_prey_full_perm_condition}) that system (\ref{sys:LVPredatorPreya}-\ref{sys:LVPredatorPrey}) cannot be permanent when pesticide treatment is detrimental to both parasitoid and SWD unless the unpulsed system is also permanent. In other words, the permanence of the untreated system dictates whether permanence can be achieved in the treated system. This is a departure from the earlier maintenance chemotherapy example, where we saw that single-species exclusion behavior in the untreated case could be transformed into permanent type behavior.


\begin{figure}
    \centering
    \includegraphics[width=0.75\linewidth]{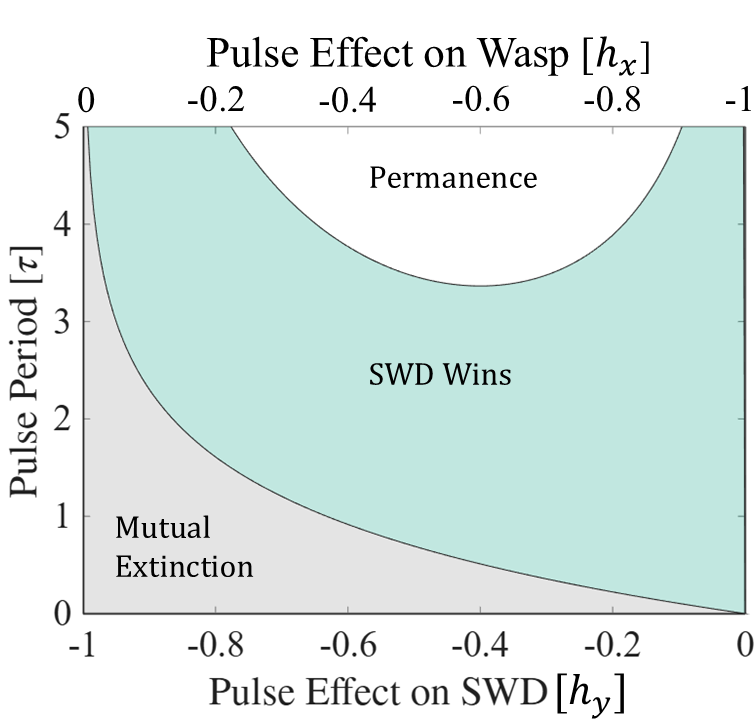}
    \caption{
    Bifurcation plot of SWD-wasp parasitoid system (\ref{sys:LVPredatorPrey}) dynamics varying pulse period length and treatment specificity, where $h_y + h_x = -1$. The unpulsed system parameters are $a_y=a_x = 1$, $b_{yy}=\frac{1}{3}$, $b_{yx} = \frac{1}{2}$, and $b_{xy}=-\frac{1}{2}$.
    }
    \label{fig:predator_prey_phase_bifurcation}
\end{figure}

To examine the extent to which this correlation holds as the pulse effect becomes more or less specific to SWD, we construct a bifurcation plot comparing choices of $h_x, h_y \leq -1$ over various pulse periods using condition (\ref{eqxn:predator_prey_full_perm_condition}). Here, we assumed $h_y+h_x=1$ to mimic that some pesticides have differential impacts on different species. 

Interestingly, as in the unpulsed predator-prey system, we see that when pulsed permanence is possible there is an optimal range for biocontrol applications in which the pesticide is neither too effective against wasp nor SWD. This suggests that treating SWD too strongly may lead to tradeoffs in the efficacy of biocontrol. 
\cite{heimpel_biocontrol_overview_2018}.

\section{Discussion}\label{sect:discussion}
Motivated by the problem of controlling harmful populations via pulsed interventions (such as in insecticide sprays in agricultural settings or chemo-therapies in cancer), 
in the present paper we examine  the permanence and robust permanence of an n-species periodically impulsive population dynamic model. In particular, we first provide  a sufficient condition for determining a permanent system, i.e., that all $n$ species will persist despite periodic pulsed interventions. This condition is particularly helpful due to its disentangling of the temporally discrete pulse, and the continuous growth dynamics. One particular challenge of the analysis is the nonautonomous nature of the n-species system. We overcome this challenge with a series of mappings which connect the nonautonomous $n$-species system to an n-species autonomous dynamical system via a continuous, periodic extension of the impulsive dynamics. We prove that these dynamical systems share key properties, and thus can study the permanence of one to infer permanence of the other. Permanence relies on a balance between the relative strength of the impulsive events measured against the frequency of such events, and the continuous demographics otherwise. 

Secondly, we demonstrate that the permanence property  is, under certain conditions, robust to perturbations in the continuous dynamics $(f_i(x))$, period length ($\tau$) and impulse control impact  ($h$).  This result acknowledges that our population dynamics models are imperfect descriptions of real populations and highlights that, provided our descriptions are sufficiently close, we still have permanent systems.  Specifically, in practice, parameters, such as the control efficacy ($h$) are often estimated in idealized laboratory experiments that lead to parameter values that deviate from ``the field”. The conditions we define in part 1 for permanence are also sufficient to show permanence of (appropriately defined) perturbed models. 


Since the concept arose, in the late 70s \cite{freedman1977threelevelfoodchain, gard1980persistencefoodwebs, gard1984conferencepersistence}, conditions for permanence and robustness have been determined in a variety of different classes of models, including several approaches for continuous-time dynamical systems \cite{garay_hofbauer_2003, HofbauerSchreiber2010Robust, Schreiber2000RobustPermanence, patel_schreiber_2018, HofbauerSchreiber2022InvasionGraphs}, discrete maps \cite{freedman1989persistence, salceanu2009lyapunov, kon2004permanence}, and abstract dynamical systems \cite{garay1989uniform, ButlerWaltman1986Persistence, hofbauer1989uniform} .  Additionally, there has been some key work on permanence in impulse models \cite{ballinger_persistence_of_general_ides_1997, schreiber_flowkick_2025}. Our work contributes to and extends this body of literature. 

Our work is most closely related to two particular works: \cite{ballinger_persistence_of_general_ides_1997} from 1997 and, more recently, \cite{schreiber_flowkick_2025}. 
Ballinger and Liu \cite{ballinger_persistence_of_general_ides_1997} examined the problem of permanence in impulsed population models of the form (1), but with a more flexible impulse function (where we have a linear function) and times of pulses (where we have periodic pulse times).  They provide sufficient conditions for a function to serve as a Lyapunov function to show permanence.  In other words, their approach to permanence is to find such a function. 
In our work, we use particular functions to serve as Lyapunov functions (first introduced in \cite{garay_hofbauer_2003} and then extended in \cite{patel_schreiber_2018} to write down permanence conditions. By explicitly writing a possible Lyapunov function, we make checking permanence conditions more streamlined.  More recently, Schreiber \cite{schreiber_flowkick_2025} also considered permanence in generalized versions of our equations (termed flow-kick systems by Schreiber) with similar flexibilities.  Our results have overlap with Schreiber’s, as he used a Morse decomposition and an explicit form of Lyapunov functions.  The two key distinctions are that we have a different proof that involves translating into an autonomous system, (thereby directly using previous results).  Secondly, we extend the strength of the permanence results by showing robust permanence (which was conjectured in \cite{schreiber_flowkick_2025}). 



To illustrate the utility of our results, we examine two case study examples that describe the pulsed control of harmful populations. These examples highlight how community interactions along with pulsed interventions that impact multiple community members simultaneously, may lead to nuanced predictions on the population dynamics under control.
In cases in which control of harmful populations involves pulsed interventions, we may be able to parameterize and tune either the timing or strength of interventions or both.  
Parameterization typically involves careful experimental set ups or inference from empirical data sets. For example, Mermer et al. 2020 \cite{mermer2021timing} quantified the efficacy of several classes of commercially available chemical compounds on different life stages of the SWD. Recently, in D'ovidio Long et al. 2026 \cite{d2026dynamical}, they generated a phenomenological expression for the nonlinear impact of therapies on cancer cells using data. 

Such empirical analysis can then facilitate being able to tune control parameters. For example, in the control of several classes of disease-causing macroparasites, such as lymphatic filariasis, public health agencies are devising protocols on how frequently and what proportion of affected communities to treat in preventative chemotherapies (in which treatments drugs are widely distributed to whole subsets of populations; \cite{world2025monitoring}).  In an agricultural context, growers may also control the period between insecticide spraying, as well as the level to which the chemicals are diluted (that affects our $h$ parameter).  For cases in which control is tunable, our results can be used to identify critical time periods and efficacies that would lead to permanence.

We now take a step back and ask: why would we want permanence?  Typically, in the case of harmful biological agents, such as disease-causing or pests, the objective of control is to eliminate a population. However, in some cases, this may not be feasible and some practitioners have advocated for a more palliative approach (that is, to allow harmful populations to persist but in a manageable manner; \cite{gatenby_cost_of_resistance_2009, whelan2020resistance})  For example, this approach has been suggested in cancer treatments to prevent the problematic recrudescence of disease that occurs from latent but drug resistant cancer cells that emerge from repeated exposure to chemotherapies \cite{gatenby_cost_of_resistance_2009}.  In agricultural settings, this may mean allowing for some pests to persist to also mitigate the evolution of pesticide resistance (e.g. see refuge hypothesis approaches; \cite{carriere2012large}).  In other cases, total elimination of harmful populations may not be possible due to logistical constraints (e.g. in cancer, requiring treatment beyond what is safe for an individual).  In these cases, if we are not able to eliminate a harmful population, it may be best to ensure the persistence of whole communities of populations (e.g. in cancer, healthy cells alongside cancerous cells). 

In addition to the control of harmful populations, our results can be used to understand coexistence of populations in other biological scenarios that involve pulsed dynamics.  For example, many biological species, such as some fish \cite{tao2008dynamic}, bat \cite{lunn2024kenyan},  and flying fox species \cite{mitchell2021birth}, are characterized by pulsed birth events.   On the other hand, other species are vulnerable to natural mass cohort die-offs (e.g. see \cite{smith2011ecological}), pulsed settlement/migration events (e.g. \cite{hamman2018landscape}), or pulsed resource contributions (e.g. \cite{gandhi2025flow}).
Numerous previous work have aimed at modeling these pulsed events, such as in \cite{eskola2007mechanistic, pachepsky2008between, lewis2012spreading}. Our results on permanence offer a general avenue to investigate coarse-grained properties of such systems.


Finally, we highlight some areas that build off this work in which more analysis is needed.  
In this work, we only considered linear impact of pulsed interventions.  It is natural to want to extend to nonlinear impacts as this may be more realistic in many cases (see \cite{d2026dynamical}).  We point out that work in \cite{schreiber_flowkick_2025} does allow for this more flexible pulse impact.  Secondly, our framework is entirely deterministic.  In many cases population dynamics as well as pulsed events (especially when naturally driven such as hurricanes and floods) may be better modeled as stochastic.  What are the conditions needed for the analogous notions of permanence in stochastic impulse systems and how do they compare to these? Finally, in many cases, deliberate control of harmful populations is adaptive in the sense that control protocols can change from the observation of new data in real time.  Modeling and then examining permanence in which control is adaptive is a mathematical challenge with important implications for optimizing both data collection for monitoring and updating control protocols.

\section{Appendix}

\subsection{From Impulse to Autonomous System}

In the first part of the appendix, we provide details on how we relate the Impulse ODE

\begin{equation} \label{appendix_eqn}
    \begin{split}
        \frac{dx_i}{dt} = x_i(t) f_i(x) ,\;  i = 1 \ldots n\; ,\ t \neq k\tau; k \in \N\\
        x_i(k\tau^+) = (1+h_i)x_i(k\tau^-)
    \end{split}
\end{equation}
to an autonomous system

\begin{align}
     \frac{dy_i}{dt} &= y_ig_i(y, \theta)\\
    \frac{d\theta}{dt} &=\frac{2\pi}{\tau+1} \label{auto_app}
\end{align}
with 
\begin{align}
    g_i(y, \theta) = \begin{cases}
        f_i(y), \; \theta \in [0,\frac{2\pi\tau}{\tau+1}) \\
        \ln(1+h_i), \; \theta \in [\frac{2\pi \tau}{\tau + 1},2\pi)
    \end{cases}
\end{align}

The first step is to relate (\ref{appendix_eqn}) to the nonautonomous periodic system

\begin{align}\label{app_cont_pulsed_model}
    \dot{y}_i(y,t) &= y_i(t)g_i(y,t) \\
    g_i(y, t) & = \begin{cases}
    f_i(y), & t \in [k(\tau + 1) , k(\tau + 1) + \tau] \\
        \ln(1 + h_i), & t \in [k(\tau + 1) + \tau , (k+1)(\tau + 1))
    \end{cases}    
\end{align}
where $k\in \N$.  Such a transformation was done in \cite{patel_spectral_2024}. Heuristically, this extends the time over which the pulse acts to one time unit rather than instantaneously (see Figure \ref{fig:sketch_conversion_to_cont}).  There is one-to-one correspondence between solutions of these two systems and importantly, if the permanence property holds in one system it must hold in the other. This transformation relies on the linearity of the pulse impacts. 

\begin{figure}[h]
    \centering
    \includegraphics[width=1\linewidth]{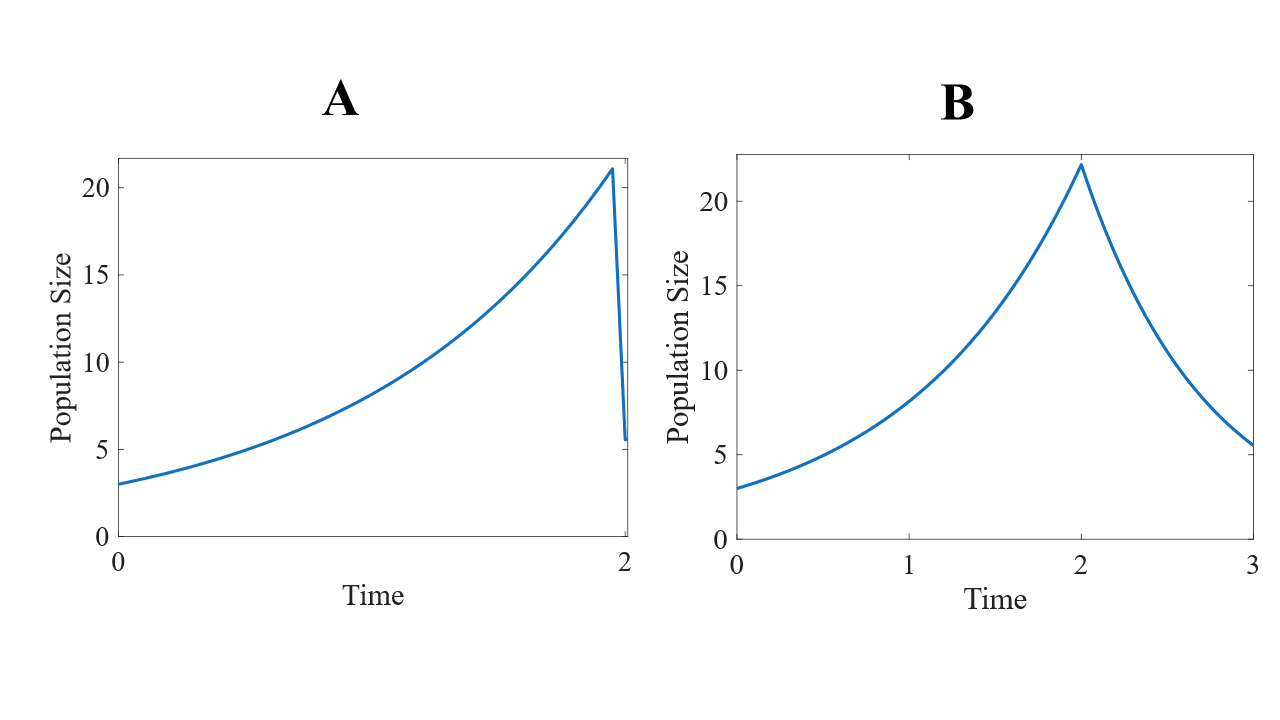}
    \caption{Example conversion between periodically pulsed system (\ref{appendix_eqn}) and periodically time-varying system (\ref{app_cont_pulsed_model}) after one period for $\dot{x(t)} = x$ and $\tau = 2$ with $h = -0.75$. On the left in (A) the original nonautonomous, pulsed solution. On the right (B) the piecewise continuous autonomous system with extended time.}
    \label{fig:sketch_conversion_to_cont}
\end{figure} 

The second step is to relate the nonautonomous system (\ref{app_cont_pulsed_model}) to an autonomous system (\ref{auto_app}). We propose a transformation of the time domain $\Delta$ to a 1-sphere $S^1$ by 
\begin{equation}\label{eqxn:time_transformation}
    \theta(t) = \frac{2\pi t}{\tau + 1}
\end{equation}
Rewriting (\ref{app_cont_pulsed_model}) under this change of variables yields the system of autonomous differential equations \begin{subequations}\label{eqxn:cont_autonomous_model}
    \begin{equation}
    \begin{cases}
        \dot{y}_i (y, \theta) = y_i(t)\overline{g}_i(y, \theta)\\
        \dot{\theta}(y, \theta) = \frac{2\pi}{\tau + 1}
    \end{cases}
    \end{equation}
    where $\overline{g}_i(y, \theta)$ is defined 
    \begin{equation}
        \overline{g}_i(y, \theta) = \begin{cases}
            f_i(y), & \theta \in \left[2\pi k, \frac{2\pi \tau}{\tau + 1}(k+1)\right)  \\
            \ln(1 + h_i), & \theta \in \left[\frac{2\pi \tau}{\tau + 1(k+1)},2 \pi(k+1) \right)
        \end{cases}    
    \end{equation}
     with $k$ here identical to that in system (\ref{app_cont_pulsed_model}) above.
\end{subequations}
This transformation relies on the periodicity inherent in the nonautonomous system. 
\subsection{Robust Permanence Result}
\begin{proof}[Proof of Theorem \ref{thm:robust_permanence}]
We begin by defining robust permanence for the autonomous system (\ref{auto}, \ref{auto2}) as in \cite{patel_schreiber_2018}. Let $Q_\phi \subset \R^n_+ \times S^1$ be a compact, forward invariant set with respect to (\ref{auto},\ref{auto2}). Then $(\tilde{g}, \tilde{\Lambda})$ is in the $(\delta, Q_\phi)$-perturbation of (\ref{auto}, \ref{auto2}) if 
    \begin{enumerate}
        \item For all $i$, $(x,\theta) \in \R^n_+ \times S^1$, and
        $$
        |\tilde{g}_i(x,\theta)-g_i(x,\theta)| < \delta , \; |\tilde{\Lambda} - \frac{2\pi}{\tau + 1}| < \delta
        $$ 
        \item $x\tilde{g}_i$ is piecewise, locally Lipschitz continuous, and
        \item the perturbed dynamical system $\tilde{\phi}$ defined by
        \begin{equation} \label{pert67}
        \begin{split}
            \frac{d\tilde{x}_i}{dt} = x\tilde{g}_i(x,\theta), \; i = 1, \ldots, n\\
            \frac{d\tilde{\theta}}{dt} = \tilde \Lambda 
        \end{split}          
        \end{equation}
        satisfies
        $$
        \bigcup_{s \in \R_+}\tilde{\phi} (Q_\phi, [0,2\pi),[s,\infty)) \subset Q_\phi
        $$
        and for all $(z,\theta) \in \R^n_+ \times S^1$, there exists a $t\in \R_+$ such that $\tilde{\phi}(z,\theta, t) \in Q_\phi$.
    \end{enumerate}
The system (\ref{auto},\ref{auto2}) is robustly permanent if there exists a $\delta, \beta > 0$ such that for all $(\tilde g, \tilde \Lambda)$ in the $(\delta, Q_\phi)$-perturbation of (\ref{auto},\ref{auto2}) 
\begin{equation}\label{eqxn:robust_perm_auto}
\frac{1}{\beta} < \liminf_{t \to \infty}{\tilde{\phi}_i(t)} \leq \limsup{\tilde{\phi}_i(t)} < \beta
\end{equation}
where $\tilde{x}(t)=(\tilde{x_1}(t), \tilde{x_2}(t), ...\tilde{x_n}(t))$ is a solution to (\ref{pert67}) with an initial condition satisfying $\tilde{x_i}(0)>0$.

In the proof of Theorem \ref{thm:main}, we showed that when condition (\ref{eqxn:persistence_inequality_pulsed}) is met, the autonomous system (\ref{auto},\ref{auto2}) generated by (\ref{eqxn:pulsed_model}) satisfies sufficient permanence conditions given in Theorem 1 in \cite{patel_schreiber_2018}. By Theorem 2 from the same paper, this implies that system (\ref{auto}, \ref{auto2}) is robustly permanent. That is, there is a $\delta , \beta> 0$ such that all $(\delta, Q_\phi)$-perturbations are permanent for a uniform choice of $\beta$. What remains to be shown is the robust permanence of (\ref{eqxn:pulsed_model}). 

For $(\tilde g, \tilde \Lambda)$ in such a ($\delta,Q_\phi$)-perturbation of (\ref{auto},\ref{auto2}), let 
    $$
    \tilde{\tau} = \frac{2\pi \tau}{(\tau + 1)}\frac{1}{\tilde{\Lambda}}, \; \tilde \Lambda \neq 0
    $$ and for $(x, \theta) \in \R^n_+ \times S^1$, $i = 1, \ldots, n$ let
    $$
    \tilde{g}_i(x,\theta) = \begin{cases}
        \tilde{f}_i(x), & \theta \in [0,\frac{2\pi \tau}{\tau + 1}]\\
        \frac{\tau}{\tilde{\tau}} \ln{(1+ \tilde{h}_i)}, & \theta \in [\frac{2\pi\tau}{\tau + 1} , 2\pi)
    \end{cases}
    $$.
    
    Notice that for $\tilde{\tau} > 0$, $\frac{2\pi}{\tau + 1}\frac{\tau}{\tilde \tau}$ is continuous with respect to $\tilde\tau$. So at the point $\tilde \tau = \tau$, there exists a $\delta_{\tau_1} > 0$ such that for all $|\tilde \tau - \tau| < \delta_{\tau_1}$, we have 
    $$
    |\frac{2\pi}{\tau + 1}\frac{\tau}{\tilde{\tau}} - \frac{2\pi}{\tau + 1}| < \delta
    $$. By similar argument, we can choose $\delta_{\tau_2}, \delta_{\tau_3}>0$ to bound 
    $$
    |\frac{\tau}{\tilde{\tau}} - 1| \leq \sqrt{\frac{\delta}{2}} -1
    $$
    for all $|\tilde{\tau} - \tau| < \delta_{\tau_2}$ and 
    $$
    |\frac{\tau}{\tilde{\tau}} - 1| \leq \frac{\delta}{2|\ln(1+h_i)|}
    $$
    for all $|\tilde{\tau} - \tau| < \delta_{\tau_3}$. For all $\tilde{h}_i > -1$, we can also choose $\delta_{h_1} > 0$ so that for all $|\tilde{h}_i -h_i| < \frac{\delta}{2}$
    $$
    |\ln{(1+\tilde{h}_i)} - \ln{(1+h_i)}| < \sqrt{\frac{\delta}{2}} 
    $$.

    Let $\tilde{\delta} = \min{\{\delta, \delta_{\tau_1}, \delta_{\tau_2}, \delta_{\tau_3}, \delta_h}\}$ and $\tilde \beta = \beta$. Then for all $i = 1, \ldots, n$ and $(x, \theta) \in \R^n_+ \times S^1$
    \begin{equation*}
        \begin{split}
            |\frac{\tau}{\tilde{\tau}}\ln{(1+\tilde{h}_i)-\ln{(1+h_i)}}| \leq |\frac{\tau}{\tilde{\tau}}\ln{(1+\tilde{h}_i)} - \frac{\tau}{\tilde{\tau}}\ln{(1+h_i)}| + |\frac{\tau}{\tilde{\tau}}\ln{(1+h_i)} - \ln{(1+h_i)}|\\
            \leq \sqrt{\frac{\delta}{2}}\sqrt{\frac{\delta}{2}} + |\frac{\tau}{\tilde{\tau}} - 1||\ln(1+h_i)| \\
            \leq \delta
        \end{split}
    \end{equation*}
    and also
     \begin{equation*}
        \begin{split}
            |\tilde{f}_i(x) - f_i(x)| < \delta \\
            |\frac{2\pi}{\tau + 1}\frac{\tau}{\tilde{\tau}} - \frac{2\pi}{\tau + 1}| < \delta
        \end{split}
    \end{equation*}
    whenever $|\tilde{f}_i(x)-f_i(x)| < \delta, \; |\tilde{\tau}-\tau| < \delta$, and $|\tilde{h}_i - h_i| < \delta$. Further, by the correspondence of (\ref{eqxn:pulsed_model}) and (\ref{auto},\ref{auto2}), if inequality (\ref{eqxn:robust_perm_auto}) holds, then for $i = 1, \ldots, n$, all solutions to (\ref{eqxn:perturbed_pulsed_model}) in the ($\tilde{ \delta}, Q$)-perturbation of (\ref{eqxn:pulsed_model}) with $\tilde{x}_i(0) > 0$ satisfy
    $$
        \frac{1}{\tilde \beta} < \liminf_{t \to \infty} \tilde{x}_i(t) < \limsup_{t \to \infty} \tilde{x}_i(t) < \tilde \beta
    $$ as well, where $\tilde{x}$ is a solution to (\ref{eqxn:perturbed_pulsed_model}).

    So robust permanence of  (\ref{auto},\ref{auto2}) implies robust permanence of (\ref{eqxn:pulsed_model}), and the conditions on (\ref{eqxn:pulsed_model}) in Theorem \ref{thm:main} are sufficient for robust permanence of (\ref{eqxn:pulsed_model}).
    
\end{proof}

\bibliographystyle{plainnat}
\bibliography{ref}

\end{document}